\documentclass[lettersize,journal]{IEEEtran}
\usepackage{amsmath,amsfonts}
\usepackage{algorithmic}
\usepackage{algorithm}
\usepackage{array}
\usepackage[caption=false,font=normalsize,labelfont=sf,textfont=sf]{subfig}
\usepackage{textcomp}
\usepackage{stfloats}
\usepackage{url}
\usepackage{verbatim}
\usepackage{graphicx}
\usepackage{cite}
\usepackage{amssymb}
\usepackage{bm} 
\usepackage{xcolor}
\usepackage{amsthm}
\def\linspread{1}
\newtheorem{proposition}{Proposition}
\begin{document}

\title{Exploiting Movable Antennas for Multi-Target Wireless  Sensing System in Multipath Environments}

\author{Xiang Chen,~\IEEEmembership{Student Member,~IEEE,}~Ming-Min Zhao,~\IEEEmembership{Senior Member,~IEEE,}~Lebin Chen,~\IEEEmembership{Student Member,~IEEE,}~Jie Xu,~\IEEEmembership{Fellow,~IEEE,} and Min-Jian Zhao,~\IEEEmembership{Member,~IEEE}
\thanks{Xiang Chen, Ming-Min Zhao, Lebin Chen and Min-Jian Zhao are with the College of Information Science and Electronic Engineering, Zhejiang University, Hangzhou 310027, China (e-mail: 12231089@zju.edu.cn; zmmblack@zju.edu.cn; 12431101@zju.edu.cn; mjzhao@zju.edu.cn).}
\thanks{Jie Xu is with the School of Science and Engineering, the Shenzhen
	Future Network of Intelligence Institute (FNii-Shenzhen), and the Guangdong Provincial Key Laboratory of Future Networks of Intelligence, The Chinese University of Hong Kong (Shenzhen), Guandong 518172, China (email: xujie@cuhk.edu.cn).}
}


\maketitle

\begin{abstract}
In this paper, we study the multi-target detection problem in a movable-antenna (MA)-enabled wireless sensing system with linear arrays, in which both the direct  line-of-sight (LoS) paths and the first-order non-LoS (NLoS) paths are explicitly considered. Unlike conventional fixed-position antenna arrays, MAs provide additional design degrees of freedom by enabling adaptive antenna positioning  to reconfigure the propagation geometry, offering great potential to enhance the sensing performance in  complex multi-target multipath scenarios. Under this setup, we first develop a cross sparsity Markov mixture prior to  derive the posterior probabilities of target locations, in which the structural correlation between the LoS and NLoS paths is effectively exploited to enhance the location estimation accuracy. Based on the derived posterior probabilities, we further analyze the angular-domain sensing performance for multi-target detection by proposing a new two-dimensional (2D)  ambiguity function as the performance metric. Next, we optimize the MA positions to suppress the sidelobe levels and narrow the mainlobe width of the proposed ambiguity function. Although the resulting problem is highly non-convex, we develop a low-complexity Dykstra-based projected gradient descent  algorithm to solve it efficiently. Finally, simulation results verify the accuracy of the proposed ambiguity function analysis, demonstrate the substantial performance gains enabled by the proposed prior model and MAs, and show that the proposed algorithm achieves performance comparable to  existing  methods with significantly lower computational complexity.
\end{abstract}

\begin{IEEEkeywords}
Wireless sensing, movable antenna (MA),
antenna position optimization, multipath, multi-target, angle estimation, cross sparsity.
\end{IEEEkeywords}

\section{Introduction}
Multiple-input multiple-output (MIMO)  wireless sensing and integrated sensing and communication (ISAC)  have attracted considerable attention to enable  emerging applications such as 6G and vehicular networks \cite{yuan2020robust,pang2024dynamic,9737357}. Meanwhile, sensing environments are becoming increasingly complex, with a growing number of targets to be detected, which poses significant challenges to accurate parameter estimation \cite{9755276}. As the number of targets increases, the transmitted sensing waveform may undergo inter-target reflections, giving rise to multipath propagation that distorts the received observations and degrades estimation performance \cite{leigsnering2014multipath}. Furthermore, an increasing number of targets inherently reduces estimation accuracy, because the increased ambiguity and inter-target interference among target returns make the estimation problem progressively more ill-conditioned \cite{5744132}. Various designs have been proposed in the literature for multi-target estimation in multipath environments, which can be broadly categorized into transmit-side designs, such as waveform, beamforming, and resource allocation, and receive-side designs, such as structural prior modeling and parameter estimation algorithm design \cite{sabet2020hybrid,visentin2018analysis,manzoni2023multipath,zheng2024detection,10777052,rong2020diffuse,xu2021mimo,chen2025cross,9676556}.

Recently, reconfigurable arrays (e.g., movable antennas (MAs) and fluid antennas (FAs)) have attracted increasing attention in wireless sensing and ISAC systems, since their capability of reconfiguring antenna positions or spatial distributions provides additional spatial degrees of freedom for sensing-performance enhancement \cite{10643473,11086422,10707252,10477314}. In \cite{10643473}, a Cram{\'e}r-Rao lower bound (CRLB)-minimizing MA placement method was developed for single-target sensing. In \cite{11086422}, MAs were introduced into ISAC systems to further improve the sensing and communication tradeoff, and an optimization framework was developed to minimize the CRLB for sensing while guaranteeing a target communication rate, thereby enhancing the system performance under explicit communication quality constraints. Meanwhile, FAs were investigated as an alternative reconfigurable-antenna architecture for ISAC in \cite{10707252,10477314}. In \cite{10707252}, the FA spatial distribution was optimized to maximize the communication SNR subject to a sensing beampattern-gain constraint, thus enabling a controlled sensing-communication tradeoff. However, the resulting optimization problem was NP-hard in general. To  better address this issue, a deep reinforcement learning (DRL)-based approach was proposed in \cite{10477314} to optimize the FA positions. Nevertheless, the above studies mainly focused on improving the sensing performance of reconfigurable-antenna-enabled MIMO systems in a single-target setting. 
	
{Compared with single-target sensing, multi-target sensing usually suffers from performance degradation due to inter-target interference and stronger parameter coupling caused by the superposition of multiple target returns \cite{mao2025movable,zhang2025direct,san2007mimo,huan2023sasa,8094946}. This issue is further compounded by the substantial increase in the dimension of the unknown parameter set with the number of targets and propagation paths, which makes CRLB or beampattern-gain-based sensing metrics much more difficult to evaluate and optimize \cite{mao2025movable}. To address this, \cite{mao2025movable} proposed a Monte Carlo simulation-based method to approximate the multi-target CRLB. However, this approach generally incurs high computational complexity and long simulation time, while the resulting approximation accuracy may still be limited.}
 Consequently, directly minimizing the CRLB for multi-target sensing is challenging. As an alternative, the ambiguity function introduced in \cite{san2007mimo} could provide a more tractable surrogate, where suppressing sidelobes  is able to mitigate inter-target leakage and mutual interference, and narrowing mainlobe  improves parameter resolution and estimation accuracy. In \cite{huan2023sasa}, a multi-objective optimization problem was formulated to simultaneously narrow the mainlobe width and suppress the sidelobe level of the ambiguity function by optimizing fixed antenna positions in a sparse array, with the aim of improving the  performance of MIMO radar systems.
 However, the  optimization method in \cite{huan2023sasa}  mainly aims to achieve a balanced reduction of the sidelobe levels over the entire angular domain, rather than targeting the specific sidelobe regions that give rise to severe inter-target interference in multi-target scenarios. Moreover, since it does not exploit the separation between the AoD and AoA, it is not applicable to multipath scenarios where different propagation paths may have mismatched AoD-AoA pairs. To better characterize the mutual interference among multiple targets, it is important to obtain coarse estimates of their locations. In multipath environments, the work \cite{chen2025cross} proposed a cross sparsity prior structure, which exploits the fact that the AoD and AoA of a first-order  non-line-of-sight (NLoS) path may coincide with those of the direct LoS paths of two different targets, thereby improving the target estimation accuracy and detection probability. However, when false alarms occur over first-order NLoS paths, the cross sparsity structure may introduce erroneous coupling to the direct LoS paths, thereby leading to inaccurate estimation of the direct LoS paths. Consequently, the existing cross-sparsity structure remains vulnerable to such erroneous coupling, which compromises its reliability for high-precision multi-target estimation in complex NLoS environments. This motivates us to incorporate a Markov prior into the cross-sparsity model to suppress erroneous path coupling and improve multi-target sensing performance.

\begin{figure}[t]
	\centering
	\includegraphics[width=2.7in]{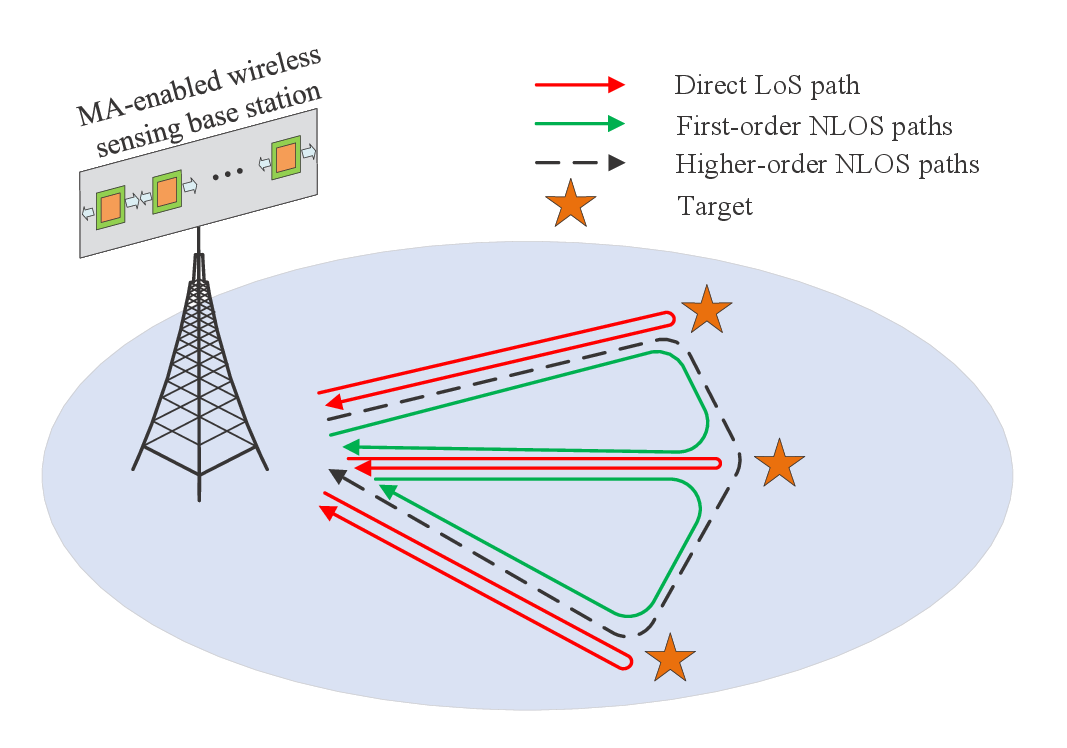}
	\caption{MA-enabled multi-target wireless sensing system.}
	\vspace{-2mm}
	\label{fig1}
	\vspace{-3mm}
\end{figure}

Motivated by the above considerations, in this paper, we investigate an MA-enabled  wireless sensing system for multi-target sensing in multipath environments, where the first-order NLoS paths are characterized by the fact that their AoAs and AoDs match the direct-path angles of different targets, as illustrated in Fig. \ref{fig1}. Moreover, we investigate how this structural property can be exploited as a powerful prior and how optimizing the MA positions can further enhance the sensing performance. In contrast to fixed arrays, MAs can adapt their antenna positions to different target geometries, which is particularly attractive for multi-target sensing because the relative separations among targets determine how they interact with the sidelobe structure of the ambiguity function, and thus directly affect the sensing performance. This adaptability provides additional spatial degrees of freedom to suppress inter-target ambiguity and improve estimation accuracy. The main contributions of this work are summarized as follows. 
 \begin{itemize}
 \item 
First, we construct a sparse representation of the multi-target sensing signal in multipath scenarios over two-dimensional (2D) angular grids by decoupling the transmit and receive angles, which enables explicit modeling of the first-order multipath components. {Building on the resulting path structure, we propose a cross sparsity Markov mixture prior. It can effectively suppress false detection of first-order NLoS paths caused by the cross sparsity assumption, and mitigate the resulting error propagation through the cross sparsity coupling, thereby improving the overall estimation accuracy.}

\item 
Second, we construct a 2D posterior-probability-based ambiguity function to characterize the multi-target sensing performance in  the considered system. Specifically, we transform the posterior probabilities of the targets' rough locations into a probability distribution over their relative position offsets, which is then used as a weight to characterize the ambiguity function sidelobe leakage across targets, thereby quantifying the resulting inter-target interference and estimation loss. Besides, we derive an explicit relationship between the mainlobe width of the proposed ambiguity function and the antenna spatial distribution. 
 
 \item 
Third, building upon the proposed ambiguity function, we formulate an optimization problem to optimize the MA positions for ambiguity function shaping, where  sidelobe levels are penalized in a posterior-weighted manner so as to more strongly suppress those regions that are most detrimental to multi-target sensing, while simultaneously narrowing the mainlobe. The resulting problem is highly nonconvex and challenging to solve efficiently. To address this issue, we develop a Dykstra-based projected gradient descent (DPGD) algorithm by exploiting the fact that the feasible set is the intersection of multiple convex sets. The proposed algorithm achieves performance comparable to the baselines while requiring substantially lower computational complexity. Simulation results further verify that the proposed cross sparsity prior and MA design can significantly improve the multi-target sensing accuracy in terms of root mean square error (RMSE).
\end{itemize}
 
 The remainder of the paper is outlined as follows. Section~II describes the system model  and formulates a sparse representation of the multi-target sensing channel. Section~III   presents the  proposed cross sparsity Markov mixture prior model. Section IV introduces the posterior-probability-based ambiguity function, and formulates the resulting optimization problem. Section~V develops the DPGD algorithm. Section~VI presents numerical results, and conclusions are drawn in Section VII. 
 
 \textit{Notations:} Scalars, vectors and matrices are respectively denoted by lower/upper case, boldface lower case and boldface upper case letters. For an arbitrary matrix $\mathbf{A}$, $\mathbf{A}^T$, $\mathbf{A}^*$ and $\mathbf{A}^H$ denote its transpose, conjugate and conjugate transpose respectively. $\|\cdot\|$ denotes the Euclidean norm of a complex vector, and $\lvert\cdot\lvert$ denotes the absolute value of a complex scalar. $\lceil \cdot \rceil$ represents the ceiling function. $\otimes$ and $\circ$  represent the Kronecker product  and Khatri-Rao (KR)
 product, respectively.    $\text{vec}(\cdot)$ represents vectorization. $\mathbf{I}$ and $\mathbf{0}$ denote an identity matrix and an all-zero vector with appropriate dimensions, respectively. $\mathbb{C}^{n\times m}$ denotes the space of $n\times m$ complex matrices. $\mathbf 1_m \in \mathbb R^{m\times1}$ denotes the all-one column vector of length $m$. $\Gamma(o; a, b)$ denotes a Gamma distribution at the variable $o$ with shape parameter $a$ and rate parameter $b$. $\mathbb{E}[\cdot]$ denotes the expectation with respect to the underlying random variables.  $\mathbb{Z}$ denotes the set of  integers. $\mathcal{CN}(c,d)$ denotes the circularly symmetric complex Gaussian distribution with mean $c$ and variance $d$. The corresponding probability density function evaluated at variable $k$ is denoted by $\Phi_{\mathcal{CN}}(k;c,d)$ $\mathcal{O}(\cdot)$ denotes the asymptotic computational complexity order.

\section{System Model}
\subsection{MA-enabled Multi-Target Sensing System}
 As shown in Fig. \ref{fig1},  an MA-enabled  wireless sensing system equipped with a colocated transmitter
and receiver is considered, which are comprised of linear arrays with $M_t$ and $M_r$ antennas, respectively. Each MA is attached to an electrical machinery, such that the interval between two
adjacent antennas can be dynamically adjusted \cite{ma2023capacity}. Assume that there are $K$ targets within the  far-field sensing region, whose azimuth angles are denoted by
$\boldsymbol{\theta}_T=[\theta_1,\theta_2,\cdots,\theta_K]^T$.
In such a multi-target scenario, the backscattered echoes may experience inter-target reflections before returning to the receive array, thereby giving rise to a multipath propagation environment. Consequently, the received echo signal comprises  both LoS and  NLoS components.
The LoS components correspond to the direct propagation path between the base station and each target, for which the AoD coincides with the AoA.
 Besides the direct LoS paths, we also consider the first-order NLoS propagation  arising from inter-target reflections, where the transmitted signal is reflected by one target and subsequently impinges on another target, so that the associated AoD and AoA are determined by different targets \cite{chen2025cross}.
 Note that higher-order NLoS paths are ignored, as they typically involve two or more reflections and thus suffer from severe attenuation.

Let $\mathbf{U}=[\mathbf{u}(1),\mathbf{u}(2),\cdots,\mathbf{u}(L)]\in\mathbb{C}^{M_t\times L}$ denote the transmitted sensing signal matrix, where $L$ is the number of snapshots and $\mathbf{u}(l)=[u_1(l),u_2(l),\cdots,u_{M_t}(l)]^{T}$ denotes the transmitted signal vector at the $l$-th snapshot \cite{zheng2024detection}. Then, the corresponding received signal matrix is expressed as $\mathbf{R}=[\mathbf{r}(1),\mathbf{r}(2),\cdots,\mathbf{r}(L)]\in\mathbb{C}^{M_r\times L}$, where
\begin{equation}\label{eq:signal_model}
	\begin{aligned}
		\mathbf{r}(l)=&\underbrace{\sum_{k=1}^{K} \alpha_k \mathbf{a}_r(\theta_k)\mathbf{a}_t^{T}(\theta_k) \mathbf{u}(l)}_{\text{Direct LoS path}}+\!\!\!\!\underbrace{\sum_{\substack{i=1,\, j=1 \\ i \neq j}}^{K} 
			\!\!\!\beta_{i,j}\mathbf{a}_r(\theta_i)\mathbf{a}_t^{T}(\theta_j) \mathbf{u}(l)}_{\text{First-order NLoS paths}}\\&+\mathbf{n}_r(l).
	\end{aligned}
	\vspace{-1mm}
\end{equation}
 In \eqref{eq:signal_model}, $\alpha_k = \zeta_k g_k$ denotes the complex gain of the direct LoS path  with target $k$, with $\zeta_k \sim \mathcal{CN}(0,1)$   modeling  the normalized random fluctuation of the target scattering coefficient and $g_k$ representing large-scale attenuation, respectively. The coefficient $\beta_{i,j} = \epsilon_{i,j} h_{i,j}$ represents the complex gain of the first-order NLoS path associated with the ordered  target pair $(i,j)$, where $i$ denotes the index of the AoD corresponding to the $i$-th target and $j$ denotes the index of the AoA corresponding to the $j$-th target. Here, $\epsilon_{i,j} \sim \mathcal{CN}(0,1)$ accounts for the normalized random fluctuation of the corresponding scattering coefficient, and $h_{i,j}$ represents the corresponding large-scale attenuation.
The  term $\mathbf{n}_r(l)$ denotes the additive white Gaussian noise (AWGN). 
 Furthermore, $\mathbf{a}_t(\theta)$ and $\mathbf{a}_r(\phi)$ denote the  transmit and receive steering vectors with AoD $\theta$ and AoA $\phi$, respectively, given by 
$
			\mathbf{a}_t(\theta)
		\triangleq
		\big[
		e^{j \frac{2\pi}{\lambda}x_{t,1} \sin(\theta)},
		\cdots,
		e^{j \frac{2\pi}{\lambda}x_{t,M_t} \sin(\theta)}
		\big]^{T}\in\mathbb{C}^{M_t}$
and
$
		\mathbf{a}_r(\phi)
	\triangleq
		\big[
		e^{j \frac{2\pi}{\lambda}x_{r,1} \sin(\phi)},
		\cdots,
		e^{j \frac{2\pi}{\lambda}x_{r,M_r} \sin(\phi)}
		\big]^{T}	\in\mathbb{C}^{M_r}$,
where $\mathbf{x}_t\triangleq[x_{t,1},x_{t,2},\ldots,x_{t,M_t}]^{T}$ and
$\mathbf{x}_r\triangleq[x_{r,1},x_{r,2},\ldots,x_{r,M_r}]^{T}$ collect the transmit- and receive-antenna positions (along the array axis), respectively. Our objective is to estimate $\theta_k, k\in\mathcal{K}\triangleq\{1,2,\cdots,K\}$ from $\mathbf{r}(l)$ in \eqref{eq:signal_model}.
\vspace{-3mm}

\subsection{Sparse Representation of the Multi-Target Sensing Channel}
As shown in  \eqref{eq:signal_model},  the received signal is the superposition of the direct LoS and the first-order NLoS components. Thus, we propose to  discretize the angular domain and cast the multi-target sensing channel into a sparse, dictionary-based form for tractable estimation.
Specifically, the  transmit and receive angular domains are uniformly discretized over
$\big[-\frac{\pi}{2},\,\frac{\pi}{2}\big]$ into $Q$ grid points, denoted by
$\bm{\bar{\theta}}_t=[\bar{\theta}_{t,1},\ldots,\bar{\theta}_{t,Q}]^{T}$ and
$\bm{\bar{\theta}}_r=[\bar{\theta}_{r,1},\ldots,\bar{\theta}_{r,Q}]^{T}$, respectively.
Moreover, to improve the estimation accuracy, we update the parameters of the direct LoS paths and the first-order NLoS paths separately.
Accordingly, we expand the angular grids to $Q_1\triangleq Q^2$ points so that the two types of paths can be refined independently in the 2D angular domain, i.e.,
$\bm{\theta}_t=\mathbf{1}_Q\otimes\bm{\bar{\theta}}_t\in\mathbb{R}^{Q_1\times1}$ and
$\bm{\theta}_r=\bm{\bar{\theta}}_r\otimes\mathbf{1}_Q\in\mathbb{R}^{Q_1\times1}$.

Let $q_{c,(i-1)K+j}$ denote the grid index that best matches the transmit-receive angle pair $(\theta_i,\theta_j)$, i.e.,
\begin{equation}
	q_{c,(i-1)K+j}
	\triangleq
	\arg\min_{q\in\mathcal{Q}^2}
	\Big(
	\big|\theta_i-\theta_{t,q}\big|
	+
	\big|\theta_j-\theta_{r,q}\big|
	\Big),
\end{equation}
where $\mathcal{Q}^2 \triangleq \{1,2,\cdots,Q_1\}$ denote the index set of all predefined transmit–receive grid pairs $(\theta_{t,q},\theta_{r,q})$, and $i,j\in\mathcal{K}$.
In particular, the case $i=j$ corresponds to the direct LoS path, while $i\neq j$ represents the first-order NLoS paths. 
We assume that $Q_1$ is sufficiently large such that different paths are associated with distinct nearest grid points. Then, we collect the indices of the nearest grid pairs into
$
\mathcal Q_c \triangleq \big\{ q_{c,(i-1)K+j} \,\big|\, i,j\in\mathcal K \big\} \subseteq \mathcal Q^2 $.
For each $q\in\mathcal{Q}_c$, there exists a unique pair of target indices $(i_q,j_q)\in\mathcal{K}\times\mathcal{K}$ such that
$
	q = q_{c,(i_q-1)K + j_q}.
$
Accordingly, for each angular-grid index $q\in\mathcal{Q}^2$ with transmit and receive grid points $(\theta_{t,q},\theta_{r,q})$, we define the corresponding angular offsets as the mismatch between the true angles $(\theta_{i_q},\theta_{j_q})$ and their nearest grid points, i.e.,
\begin{equation}
	\Delta\theta_{t,q} \triangleq
	\begin{cases}
		\theta_{i_q} - \theta_{t,q}, & \text{if } q \in \mathcal Q_c,\\[4pt]
		0,                            & \text{if } q \notin \mathcal Q_c,
	\end{cases}
\end{equation}
\begin{equation}
	\Delta\theta_{r,q} \triangleq
	\begin{cases}
		\theta_{j_q} - \theta_{r,q}, & \text{if } q \in \mathcal Q_c,\\[4pt]
		0,                            & \text{if } q \notin \mathcal Q_c.
	\end{cases}
\end{equation}
We further define
$\bm{\Delta\theta}_t \triangleq [\Delta\theta_{t,1},\ldots,\Delta\theta_{t,Q_1}]^{T}$ and
$\bm{\Delta\theta}_r \triangleq [\Delta\theta_{r,1},\ldots,\Delta\theta_{r,Q_1}]^{T}$
as the transmit- and receive-side angular offset vectors, respectively.
These offsets yield a refined angle representation by expressing each true AoD/AoA as the sum of its nominal grid angle and the corresponding offset, and they  play a pivotal role in enabling super-resolution sensing in the proposed framework.

Based on the above definitions, we construct the angle-dependent transmit and receive dictionaries associated with the expanded grids $\bm{\theta}_t$ and $\bm{\theta}_r$ as
$
	\mathbf{A}_t(\bm{\Delta\theta}_t)
	\triangleq
	\big[\mathbf{a}_t(\theta_{t,1}+\Delta\theta_{t,1}),\, \ldots,\, \mathbf{a}_t(\theta_{t,Q_1}+\Delta\theta_{t,Q_1})\big]
	\in \mathbb{C}^{M_t\times Q_1},
$
and
$
	\mathbf{A}_r(\bm{\Delta\theta}_r)
	\triangleq
	\big[\mathbf{a}_r(\theta_{r,1}+\Delta\theta_{r,1}),\, \ldots,\, \mathbf{a}_r(\theta_{r,Q_1}+\Delta\theta_{r,Q_1})\big]
	\in \mathbb{C}^{M_r\times Q_1}.
$
Then, the received signal model in \eqref{eq:signal_model} can be rewritten in a sparse form as
\begin{equation}
	\mathbf{R} = \mathbf{H}\mathbf{U} + \mathbf{N}_r,
\end{equation}
where
\begin{equation}\label{eq:0419_1}
	\mathbf{H}
	\triangleq
	\mathbf{A}_r(\bm{\Delta\theta}_r)\,\mathrm{diag}(\mathbf{x})\,\mathbf{A}_t^{T}(\bm{\Delta\theta}_t)
	\in \mathbb{C}^{M_r\times M_t},.
\end{equation}
In \eqref{eq:0419_1},  $\mathbf{x}\in\mathbb{C}^{Q_1\times 1}$ is the angle-domain sparse sensing channel vector, whose $q$-th entry $x_q$ denotes the complex coefficient of the sensing path associated with the $q$-th AoA-AoD grid pair.
$\mathbf{N}_r=[\mathbf{n}_r(1),\mathbf{n}_r(2),\ldots,\mathbf{n}_r(L)]\in\mathbb{C}^{M_r\times L}$ denotes the received noise matrix.
At the receiver of the wireless sensing system, a matched filter is applied to the transmitted waveform, which yields $
	\mathbf{Y} = \mathbf{R}\mathbf{U}^H.$
For subsequent processing, we vectorize $\mathbf{Y}$ as
\begin{equation}\label{eq:1014_1}
	\mathbf{y}
	\triangleq
	\mathrm{vec}(\mathbf{Y})
	=
	\mathbf{F}\!\left(\bm{\Delta\theta}_t,\bm{\Delta\theta}_r\right)\mathbf{x}
	+
	\mathbf{n}
	\in\mathbb{C}^{M_tM_r\times 1},
\end{equation}
where $\mathbf{n}\triangleq \mathrm{vec}(\mathbf{N}_r\mathbf{U}^H)$ and the sensing measurement matrix is given by
\begin{equation}\label{eq:1014_2}
	\mathbf{F}\!\left(\bm{\Delta\theta}_t,\bm{\Delta\theta}_r\right)
	=
	\big((\mathbf{U}\mathbf{U}^H)^T\otimes \mathbf{I}_{M_r}\big)
	\big(\mathbf{A}_t(\bm{\Delta\theta}_t)\circ \mathbf{A}_r(\bm{\Delta\theta}_r)\big).
\end{equation}

 Our objective is to improve the multi-target estimation accuracy by optimizing the antenna positions. This is challenging because, in multi-target multipath scenarios, the increased number of unknown parameters $\{\bm{\Delta\theta}_t,\bm{\Delta\theta}_r,\mathbf{x}\}$ generally makes a closed-form CRLB unavailable, whereas the conventional 2D ambiguity function cannot accurately capture the effect of inter-target sidelobe interference on estimation performance. To address this issue, we first propose a cross-sparsity Markov mixture prior model and incorporate it into the fast turbo variational Bayesian inference (SF-TVBI) framework proposed in \cite{chen2025cross} to obtain coarse estimates of multiple target locations under a uniform linear array (ULA). Based on the posterior information obtained from the coarse estimates, we develop a new 2D posterior-probability-based ambiguity function to characterize the mutual sidelobe interference among targets, and further formulate an MA position optimization problem to determine favorable antenna locations. Finally, the target parameters are refined at the optimized antenna positions using the SF-TVBI algorithm under the proposed prior model. Fig.~\ref{figx1} summarizes the multi-target estimation procedure for the considered MA-enabled wireless sensing system.

\begin{figure}[t]
	\centering
	\includegraphics[width=3.5in]{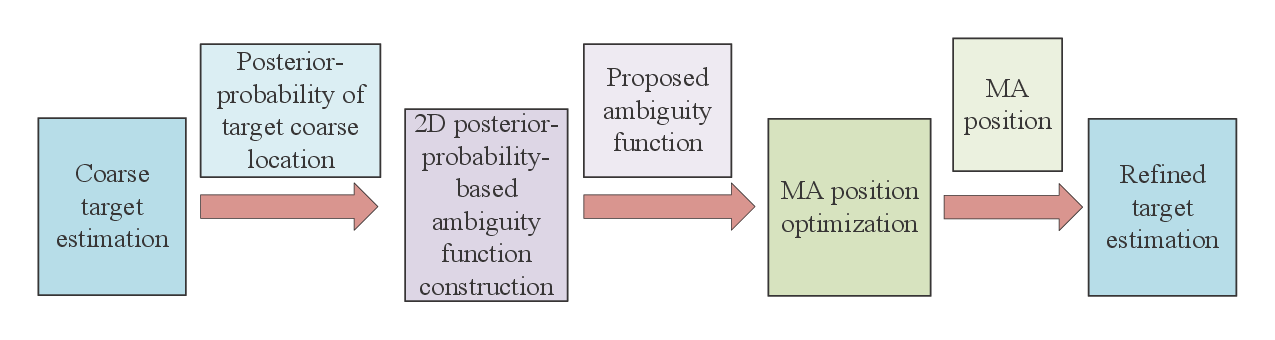}
	\vspace{-3mm}
	\caption{Illustration of the multi-target estimation procedure for the MA-enabled  wireless sensing system.}
	\vspace{-3mm}
	\label{figx1}
\end{figure}

\section{Proposed Cross-Sparsity Markov Mixture Prior Model}

For multi-target sensing in multipath environments, the work \cite{chen2025cross} proposed to exploit   the structural correlation between the first-order NLoS paths and the direct LoS paths under the sparse channel model in \eqref{eq:1014_1}. Specifically, since the AoD and AoA of a first-order NLoS path coincide with the direct-path angles of two different targets, the presence of such an NLoS path can provide additional evidence for the existence of the associated direct LoS paths. To exploit this structural property for improving multi-target sensing performance, \cite{chen2025cross} introduced  a three-layer hierarchical shrinkage (3LHS) prior to modeling the structured sparsity of the sparse channel vector $\mathbf{x}$.\footnote{Under this 3LHS prior, \cite{chen2025cross} estimates the target parameters via the SF-TVBI algorithm, which alternately updates the approximate posteriors of the latent variables in the 3LHS model and the angular offset parameters $\{\bm{\Delta\theta}_t,\bm{\Delta\theta}_r\}$.} Specifically, a precision vector $\boldsymbol{\rho}\triangleq[\rho_{1},\rho_{2},\ldots,\rho_{Q_1}]^{T}$ is introduced to characterize the scale variation of the entries in the sparse channel vector $\mathbf{x}$, where $1/\rho_q$ corresponds to the variance of $x_q$. In addition, a binary support vector $\mathbf{s}\triangleq[s_1,s_2,\ldots,s_{Q_1}]^{T}$ is used to indicate whether each entry of $\mathbf{x}$ is active, where $s_q=1$ means that $x_q$ is active and $s_q=0$ otherwise. With these latent variables, the joint distribution is written as
\begin{equation}\label{eq:0112_1}
	p(\mathbf{x},\boldsymbol{\rho},\mathbf{s})
	= p(\mathbf{x}|\boldsymbol{\rho})\,p(\boldsymbol{\rho}|\mathbf{s})\,p(\mathbf{s}).
\end{equation}
The cross sparsity is embedded in $p(\mathbf{s})$, so that the coupling between the first-order NLoS paths and the direct LoS paths can be incorporated into the prior modeling of the sparse channel. However, since this mechanism enhances the prior existence probability of the direct LoS paths according to the detected first-order NLoS paths, false alarms associated with the latter may also be propagated to the former, thereby degrading the coarse estimation accuracy. To address this issue, we further propose a cross-sparsity Markov mixture prior model, which aims to preserve the beneficial cross-sparsity correlation while improving robustness against false alarms.

Before introducing the proposed prior structure, we first complete the 3LHS model. Specifically, the sparse vector $\mathbf{x}$ is modeled as a zero-mean circularly symmetric complex Gaussian random vector with covariance $\mathrm{diag}(\boldsymbol{\rho}^{-1})$, i.e.,
\begin{equation}
	p(\mathbf{x}|\boldsymbol{\rho})
	= \prod_{q=1}^{Q_1} p(x_q|\rho_q)
	= \prod_{q=1}^{Q_1} \Phi_{\mathcal{CN}}\!\left(x_q;0,1/\rho_q\right).
\end{equation}
Such a conditional Gaussian modeling of $\mathbf{x}$, together with the hyperprior imposed on the precision variables, leads to a heavy-tailed marginal distribution after marginalizing over $\boldsymbol{\rho}$, thereby making it suitable for sparse signal modeling \cite{liu2020robust}. Since $p(\mathbf{x}|\boldsymbol{\rho})$ is conditionally complex Gaussian, a conjugate Gamma prior is assigned to each precision variable so as to enable closed-form updates and promote sparsity through shrinkage. Accordingly, the conditional prior of $\boldsymbol{\rho}$ is defined as a support-dependent Gamma mixture
\begin{equation}
	p(\boldsymbol{\rho}|\mathbf{s})=
	\prod_{q=1}^{Q_1}
	\Big(\Gamma(\rho_q;a_q,b_q)\Big)^{s_q}
	\Big(\Gamma(\rho_q;\bar a_q,\bar b_q)\Big)^{1-s_q}.
\end{equation}
Here, the parameters $(a_q,b_q)$ are chosen such that $
	\frac{a_q}{b_q}=\mathcal{O}(1)$,
which assigns a moderate prior precision to active coefficients and thus preserves their flexibility. 
In contrast, $(\bar a_q,\bar b_q)$ are selected such that 
$\frac{\bar a_q}{\bar b_q}\gg 1$.
The resulting large prior precision drives the variance $1/\rho_q$ toward zero, thereby strongly suppressing inactive components. The noise precision $\gamma$ is also modeled as a Gamma random variable, i.e.,
$p(\gamma)=\Gamma(\gamma;c,d)$, where $c$ and $d$ are hyperparameters.

\begin{figure}[t]
	\centering
	\includegraphics[width=3.0in]{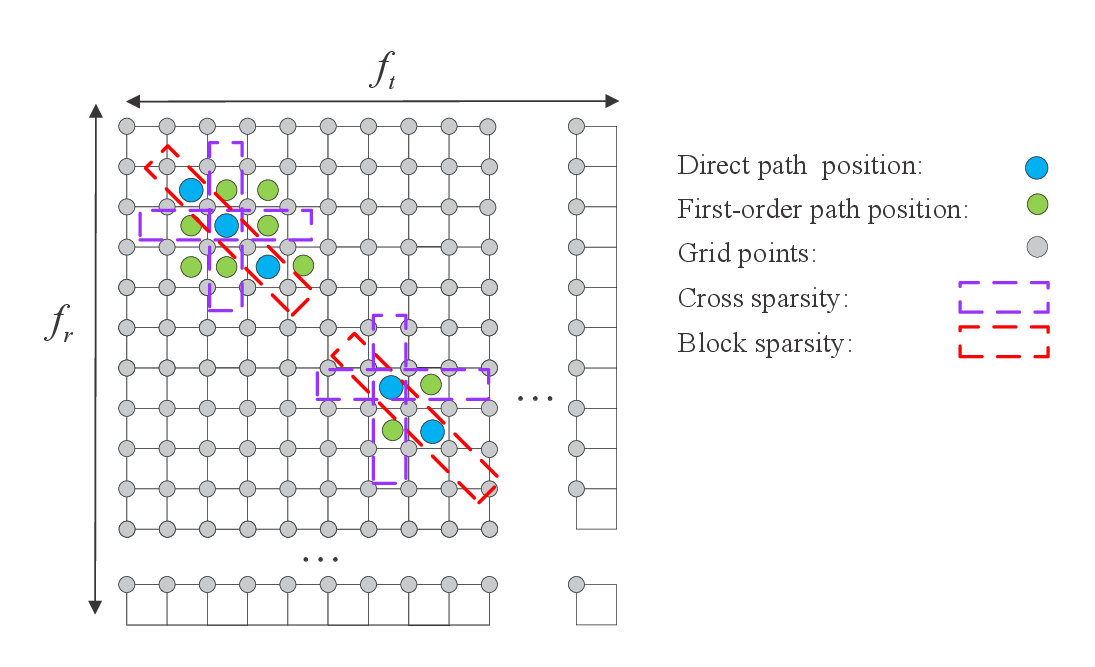}
	\caption{Illustration of proposed block and cross mixture sparsity.}
	\label{fig2}
	\vspace{-3mm}
\end{figure}
Among the three factors in \eqref{eq:0112_1}, the prior $p(\mathbf{s})$ is further  proposed to be modeled as  a proposed cross-sparsity Markov mixture prior, which plays the key role in capturing the structured sparsity pattern and the dependency among the support states in $\mathbf{s}$. As illustrated in Fig.~\ref{fig2}, in the AoD/AoA coefficient grid associated with the sensing dictionary, the direct LoS paths tend to form a block-sparse structure around the main diagonal due to spatial correlation across neighboring angular cells. In addition, inter-target first-order multipath produces off-diagonal components whose AoD or AoA coincide with those of different direct LoS paths, thereby inducing a cross-structured dependency between the off-diagonal first-order NLoS paths and the diagonal direct LoS paths. Therefore, $p(\mathbf{s})$ should jointly characterize the diagonal block sparsity of the direct LoS paths and the cross sparsity coupling introduced by the first-order NLoS paths. Accordingly, we define the diagonal indices as
	$q_i \triangleq (i-1)Q+i,\quad i\in\mathcal{Q}\triangleq\{1,\ldots,Q\}$,
such that $s_{q_i}$ denotes the support variable associated with the $i$-th direct LoS path grid on the main diagonal.
Let
$
	\mathbf{s}_{\mathrm d}\triangleq[s_{q_1},\ldots,s_{q_Q}]^{T}$
collect all diagonal support variables, and  $\mathbf{s}_{\mathrm f}$ collect the support variables on the off-diagonal grids that correspond to the first-order NLoS paths.
We model the overall support prior as
\begin{equation}
	p(\mathbf{s};\boldsymbol{\eta},\boldsymbol{\omega})
	=
	\frac{1}{Z(\boldsymbol{\eta},\boldsymbol{\omega})}\,
	\psi_{\mathrm{block}}(\mathbf{s}_{\mathrm d};\boldsymbol{\eta})\,
	\psi_{\mathrm{cross}}(\mathbf{s}_{\mathrm d},\mathbf{s}_{\mathrm f};\boldsymbol{\omega}),
\end{equation}
where $Z(\boldsymbol{\eta},\boldsymbol{\omega})$ is the partition function that ensures $\sum_{\mathbf{s}} p(\mathbf{s};\boldsymbol{\eta},\boldsymbol{\omega})=1$, and $\psi_{\mathrm{block}}(\cdot)$ and $\psi_{\mathrm{cross}}(\cdot)$ denote the potential functions associated with the block and cross sparsity, respectively. Note that $\psi_{\mathrm{block}}(\cdot)$ and $\psi_{\mathrm{cross}}(\cdot)$ are nonnegative potentials rather than normalized probability distributions, and the overall normalization of $p(\mathbf{s};\boldsymbol{\eta},\boldsymbol{\omega})$ is enforced by $Z(\boldsymbol{\eta},\boldsymbol{\omega})$. Their specific forms are provided below.

\textit{1) $\psi_{\mathrm{block}}(\cdot)$ (block sparsity):}
To capture the blockwise  occurrence of the direct LoS paths in the angular domain, we model
$\psi_{\mathrm{block}}(\mathbf{s}_{\mathrm{d}};\boldsymbol{\eta})$
as a first-order Markov chain over the direct LoS path support indices, i.e.,
\begin{equation}\label{eq:diag_mc_potential}
	\psi_{\mathrm{block}}(\mathbf{s}_{\mathrm{d}};\boldsymbol{\eta})
	= p(s_{q_1}) \prod_{i=2}^{Q} p(s_{q_i}| s_{q_{i-1}}),
\end{equation}
where $\boldsymbol{\eta}$ collects the parameters of the initial distribution and the transition probabilities of the Markov chain.
For the binary state $s_{q_i}\in\{0,1\}$, the Markov-chain potential admits the standard parameterization used in the literature.
Specifically, the initial distribution and the transition probabilities are respectively  given by
\begin{equation}
	p(s_{q_1})
	= \pi_0^{\,1-s_{q_1}}\;\pi_1^{\,s_{q_1}},
\end{equation}
\begin{equation}
	p(s_{q_i}| s_{q_{i-1}})
	=
	\begin{cases}
		(1-p_{01})^{1-s_{q_i}}\, p_{01}^{s_{q_i}}, & s_{q_{i-1}} = 0,\\[1mm]
		p_{10}^{1-s_{q_i}}\, (1-p_{10})^{s_{q_i}}, & s_{q_{i-1}} = 1,
	\end{cases}
\end{equation}
where $i=2,\ldots,Q$, $\pi_0 \triangleq \Pr\{s_{q_1}=0\}$ and $\pi_1 \triangleq \Pr\{s_{q_1}=1\}$ denote the initial-state probabilities with $\pi_0+\pi_1=1$.
Moreover, $p_{01}\triangleq\Pr\{s_{q_i}=1| s_{q_{i-1}}=0\}$ and $p_{10}\triangleq\Pr\{s_{q_i}=0| s_{q_{i-1}}=1\}$ are the transition probabilities.
Here, $p_{01}$ controls the  {activation} from an inactive state to an active state, whereas $p_{10}$ controls the  {deactivation} from an active state to an inactive state.
The parameters $\pi_1$, $\pi_0$, $p_{01}$, and $p_{10}$ control the activation level and the persistence of the Markov chain across adjacent indices. As a result, the active states along the direct LoS path indices tend to appear in contiguous runs, which is consistent with the block sparsity structure of the direct LoS paths in the angular domain. Compared with the conventional independent Bernoulli prior, the proposed Markov-chain potential explicitly introduces {local dependency} at the graphical-model level. 

\textit{2)  $\psi_{\mathrm{cross}}(\mathbf{s}_{\mathrm d},\mathbf{s}_{\mathrm f};\boldsymbol{\omega})$ (cross sparsity):}
For the first-order NLoS paths,   let
$
s_{(i-1)Q+j},\, i,j\in\mathcal{Q},\, i\neq j$,
denote the support variable on the grid corresponding to the $i$-th transmit angle and the $j$-th receive angle.
This variable exhibits a structured dependency with the direct LoS path support $s_{q_i}$ that shares the same transmit angle,
as well as the direct LoS path support $s_{q_j}$ that shares the same receive angle.
Accordingly, we define the cross potential as
\begin{equation}\label{eq:cross_potential}
	\begin{aligned}
		&	\psi_{\mathrm{cross}}(\mathbf{s}_{\mathrm{d}},\mathbf{s}_{\mathrm{f}};\boldsymbol{\omega})
		\\&= \prod_{i\neq j}
		\varphi^{(t)}\big(s_{(i-1)Q+j}, s_{q_i}\big)\,
		\varphi^{(r)}\big(s_{(i-1)Q+j}, s_{q_j}\big),
	\end{aligned}
\end{equation}
where
$
	\varphi^{(t)}\big(s_{(i-1)Q+j}, s_{q_i}\big)
	\triangleq e^{\omega^{(t)}_{i,j}\, s_{(i-1)Q+j}\, s_{q_i}}$
captures the interaction between the first-order NLoS path support $s_{(i-1)Q+j}$ and the direct LoS path support $s_{q_i}$ that share the same transmit-angle index, and
$
	\varphi^{(r)}\big(s_{(i-1)Q+j}, s_{q_j}\big)
	\triangleq e^{\omega^{(r)}_{i,j}\, s_{(i-1)Q+j}\, s_{q_j}}$
captures the interaction between the first-order NLoS path support $s_{(i-1)Q+j}$ and the direct LoS path support $s_{q_j}$ that share the same receive-angle index.
The parameters $\omega^{(t)}_{i,j}$ and $\omega^{(r)}_{i,j}$ quantify the corresponding coupling strengths along the transmit-angle and receive-angle dimensions, respectively.
 A positive interaction parameter (e.g., $\omega^{(t)}_{i,j}>0$ or $\omega^{(r)}_{i,j}>0$) encourages
the co-activation of a first-order NLoS path and its associated direct LoS path, i.e., when $s_{(i-1)Q+j}=1$, it increases the prior
tendency that the corresponding $s_{q_i}$ or $s_{q_j}$ is active. 
 
In summary, the prior $p(\mathbf{s};\boldsymbol{\eta},\boldsymbol{\omega})$ simultaneously encodes the block sparsity of the
direct LoS paths and the cross sparsity of the first-order NLoS paths relative to the direct LoS paths, via the main-diagonal Markov-chain
potential $\psi_{\mathrm{diag}}$ and the row/column interaction potential $\psi_{\mathrm{cross}}$, respectively.
Specifically, the Markov-chain structure helps suppress discrete fluctuations of the direct LoS path support under low SNR and enhances
the stability of coherent activations, while the cross-coupling structure characterizes the dependency between direct and first-order
paths along the shared transmit/receive-angle dimensions. The joint modeling of these two structures yields a physically interpretable
and structurally rich support prior for subsequent Bayesian inference.

Based on the above hierarchical prior model, the joint distribution of all associated random variables can be written as
\begin{equation}\label{eq:0822_4}
	\begin{aligned}
		&p(\mathbf{y},\mathbf{x},\boldsymbol{\rho},\mathbf{s},\gamma;\boldsymbol{\xi})
		\\&= p(\mathbf{y}|\mathbf{x},\gamma;\boldsymbol{\Delta{\theta}}_t,\boldsymbol{\Delta{\theta}}_r)\,
		p(\mathbf{x}|\boldsymbol{\rho})\,
		p(\boldsymbol{\rho}|\mathbf{s})\,
		p(\mathbf{s};\boldsymbol{\omega})\,
		p(\gamma),
	\end{aligned}
\end{equation}
where the parameter set is denoted by
$
\boldsymbol{\xi}\triangleq\{\boldsymbol{\Delta\theta}_t,\boldsymbol{\Delta\theta}_r, \boldsymbol{\omega}\}$,
and
$
p(\mathbf{y}|\mathbf{x},\gamma; \boldsymbol{\Delta{\theta}}_t,\boldsymbol{\Delta{\theta}}_r)
= \mathcal{CN}\!\left({\mathbf{F}}(\boldsymbol{\Delta{\theta}}_t,\boldsymbol{\Delta{\theta}}_r)\mathbf{x}, \gamma^{-1}\mathbf{I}\right)$.

Given the observation vector $\mathbf{y}$ and the parameter set $\boldsymbol{\xi}$, the following
marginal posteriors can be inferred via Bayesian inference, i.e,
$
p(\mathbf{x}|\mathbf{y}; \boldsymbol{\xi}), p(\boldsymbol{\rho}|\mathbf{y}; \boldsymbol{\xi}),
p(\mathbf{s}|\mathbf{y}; \boldsymbol{\xi}),
p(\gamma|\mathbf{y}; \boldsymbol{\xi})$, which can be derived by the E-step of the  SF-TVBI algorithm proposed in \cite{chen2025cross}.
Furthermore, the optimal parameter set $\boldsymbol{\xi}^*$ can be obtained by solving the following maximum-likelihood problem, i.e,
\begin{equation}\label{eq:0825_1}
	\boldsymbol{\xi}^*
	= \arg\max_{ \boldsymbol{\xi}}
	\ln p(\mathbf{y}| \boldsymbol{\xi}),
\end{equation}
which can be obtained by the M-step of SF-TVBI.\footnote{Since this work focuses on the gains in multi-target estimation accuracy enabled by the proposed hybrid prior structure and MA position optimization, rather than on the inference algorithm design itself, we simply  adopt the SF-TVBI algorithm in \cite{chen2025cross} for target sensing.} In particular, the proposed 2D posterior-probability-based ambiguity function is constructed based on the posterior $p(\mathbf{s}| \mathbf{y};\boldsymbol{\xi})$ obtained in the coarse estimation stage via the E-step of the SF-TVBI algorithm, where $p(\mathbf{s}| \mathbf{y};\boldsymbol{\xi})$ characterizes the probabilities of the targets being located on different grids over the $Q\times Q$ target plane. These posterior probabilities are then mapped to the offset grid to weight different sidelobe regions for subsequent MA position optimization.

\section{Proposed Ambiguity Function and Problem Formulation}
In this section, we propose a  2D posterior-probability-based ambiguity function to better characterize the sensing performance of the MA-enabled  wireless sensing system in multi-target multipath environments. Specifically, by incorporating the proposed prior model into the SF-TVBI framework \cite{chen2025cross}, we first perform coarse target estimation and infer the posterior probabilities of the target locations. Based on these posterior probabilities, we further evaluate the posterior probabilities associated with each target offset and use them to weight the sidelobes of the conventional 2D ambiguity function, yielding the final formulation. On this basis, we formulate an MA position optimization problem to adjust the antenna locations so as to narrow the mainlobe and suppress the sidelobe levels at competing offset locations, thus improving the multi-target sensing performance.

\subsection{2D Ambiguity Function and Mainlobe Width Analysis}
We first briefly discuss the conventional  one-dimensional (1D)  ambiguity function in the angular domain, defined as the cosine similarity between the steering vectors associated with two directions \cite{eric1998ambiguity}. Specifically, it  measures the array’s ability to discriminate between targets arriving from two directions, denoted by
\begin{equation}
	\begin{aligned}
		\mathrm{AAF}(\theta_i,\theta_j)
		&= \frac{\mathbf{a}(\theta_j)^{H}\mathbf{a}(\theta_i)}
		{\left\|\mathbf{a}(\theta_j)\right\|\left\|\mathbf{a}(\theta_i)\right\|} \\
		&= \frac{1}{M_tM_r}\sum_{m=1}^{M_t}\sum_{n=1}^{M_r}
		e^{j\frac{2\pi}{\lambda}\left(x_{t,m}+x_{r,n}\right)\left(\sin\theta_i-\sin\theta_j\right)},
	\end{aligned}
\end{equation}
where  $\mathbf{a}(\theta)\triangleq \mathbf{a}_t(\theta)\otimes\mathbf{a}_r(\theta)$ represents the  virtual array steering
vector, $\theta_i$ and $\theta_j$ represent  the reference angle and test angle, respectively. It is noteworthy that this 1D ambiguity function fails to capture the transmit and receive angle separation for multi-target sensing in multi-path scenarios, because  it only accounts for the case where the transmit and receive angles are identical. Therefore, following the basic definition of the ambiguity function \cite{san2007mimo}, we develop a 2D ambiguity function as
\begin{equation}\label{eq:1130}
	\begin{aligned}
		&\chi({\theta_t,\theta_r,\theta'_t,\theta'_r})\\&\triangleq\frac{(\mathbf{a}_t(\theta_t)\otimes\mathbf{a}_r(\theta_r))^H(\mathbf{a}_t(\theta'_t)\otimes\mathbf{a}_r(\theta'_r))}{\|\mathbf{a}_t(\theta_t)\otimes\mathbf{a}_r(\theta_r)\|\|\mathbf{a}_t(\theta'_t)\otimes\mathbf{a}_r(\theta'_r)\|}\\
		&=\frac{1}{M_tM_r}	\sum_{m=1}^{M_t}
		e^{j \alpha_{t,m}\big(\sin\theta'_t - \sin\theta_t\big)}
		\sum_{n=1}^{M_r}
		e^{j\alpha_{r,n}\big(\sin\theta'_r - \sin\theta_r\big)}	,\\
	\end{aligned}
\end{equation}
where $\alpha_{t,m}\triangleq\frac{2\pi x_{t,m}}{\lambda}$ and $\alpha_{r,n}\triangleq\frac{2\pi x_{r,n}}{\lambda}$. $\chi({\theta_t,\theta_r,\theta'_t,\theta'_r})$ quantifies the normalized correlation between the array responses associated with two transmit-receive angle pairs in the joint angular sensing space. A larger magnitude of $\chi({\theta_t,\theta_r,\theta'_t,\theta'_r})$ corresponds to higher ambiguity and thus lower angular resolvability, whereas a smaller magnitude indicates that the two angle pairs are more distinguishable. It is observed from \eqref{eq:1130} that the ambiguity function does not depend on the four angle variables $\{\theta_t,\theta_r,\theta_t',\theta_r'\}$ individually, but only on the two differences $\sin\theta_t'-\sin\theta_t$ and $\sin\theta_r'-\sin\theta_r$. Therefore, although $\chi(\theta_t,\theta_r,\theta_t',\theta_r')$ is originally defined over two transmit-receive angle pairs, it can be equivalently characterized in a two-dimensional offset domain. To facilitate the subsequent derivation, we introduce the spatial-frequency  variables $f_t\triangleq \sin\theta_t$, $f'_t\triangleq \sin\theta'_t$, $f_r\triangleq \sin\theta_r$, and $f'_r\triangleq \sin\theta'_r$, and define the corresponding spatial-frequency  offsets as $\Delta f_t\triangleq f'_t-f_t$ and $\Delta f_r\triangleq f'_r-f_r$. Then, \eqref{eq:1130} can be equivalently rewritten as
\begin{equation}\label{eq:1131}
	\begin{aligned}
		\chi({f_t,f_r,f'_t,f'_r})
		&=
		\!\frac{1}{M_tM_r}\sum_{m=1}^{M_t}
		\!\! e^{j \alpha_{t,m}\big(f'_t- f_t\big)}
		\sum_{n=1}^{M_r}
		e^{j \alpha_{r,n}\big(f'_r -f_r\big)}
		\\&	=
	\frac{1}{M_tM_r}	\sum_{m=1}^{M_t}
		e^{j \alpha_{t,m}\Delta f_t}
		\sum_{n=1}^{M_r}
		e^{j \alpha_{r,n}\Delta f_r}.
	\end{aligned}
\end{equation}

As indicated by \eqref{eq:1131}, $\chi({f_t,f_r,f'_t,f'_r})$ is given by the product of two sums of complex exponentials. Therefore, establishing explicit mathematical relationships between the antenna distribution and overall ambiguity function characteristics, particularly the sidelobe levels, is generally difficult. In contrast, the mainlobe width is more amenable to theoretical characterization, as the ambiguity function can be locally approximated in the vicinity of $\Delta f_t = 0$ and $\Delta f_r = 0$. This observation motivates us to first focus  on the relationship between the mainlobe width and the antenna distribution. The following proposition characterizes the geometry of the mainlobe region of the considered 2D ambiguity function.
\begin{proposition}\label{prop:mainlobe_ellipse}
Under the local mainlobe approximation, the mainlobe region of the 2D ambiguity function, defined by the $6\,\mathrm{dB}$ power contour of  $\lvert\chi(f_t,f_r,f'_t,f'_r)\rvert^2$ relative to its peak, can be  approximately characterized by the ellipse
	\begin{equation}\label{eq:prop_mainlobe_ellipse}
		\frac{\Delta f_t^2}{a_t^2}+\frac{\Delta f_r^2}{a_r^2}=1,
	\end{equation}
	where $\Delta f_t\triangleq f_t-f'_t$, $\Delta f_r\triangleq f_r-f'_r$, $a_t^2\triangleq \frac{1-\rho_6}{\sigma_t^2}$ and $a_r^2\triangleq \frac{1-\rho_6}{\sigma_r^2}$, with $\rho_6\triangleq 10^{-0.6}$, $\sigma_t^2\triangleq \mu_{2,t}-\mu_{1,t}^2$, and $\sigma_r^2\triangleq \mu_{2,r}-\mu_{1,r}^2$. Here, $\mu_{1,t}\triangleq \frac{1}{M_t}\sum_{m=1}^{M_t}\alpha_{t,m}$, $\mu_{2,t}\triangleq \frac{1}{M_t}\sum_{m=1}^{M_t}\alpha_{t,m}^2$, $\mu_{1,r}\triangleq \frac{1}{M_r}\sum_{n=1}^{M_r}\alpha_{r,n}$, and $\mu_{2,r}\triangleq \frac{1}{M_r}\sum_{n=1}^{M_r}\alpha_{r,n}^2$.
	Accordingly, the mainlobe width is given by the  minor-axis length of the ellipse, i.e.,
	\begin{equation}\label{eq:prop_mainlobe_width}
		\begin{aligned}
					\mathrm{W}_{\mathrm{main}}
			&=
			\min\{2a_t,\,2a_r\}
			=
			\min\left\{
			2\sqrt{\frac{1-\rho_6}{\sigma_t^2}},
			2\sqrt{\frac{1-\rho_6}{\sigma_r^2}}
			\right\}.
		\end{aligned}
	\end{equation}
\end{proposition}

\begin{proof}
	The proof is provided in Appendix~A.
\end{proof}
 It is worth noting that the proposed 2D posterior-probability-based ambiguity function only introduces posterior-probability weighting to the sidelobe terms, while leaving the local mainlobe structure unchanged. Therefore, its mainlobe width is exactly the same as that of the conventional 2D ambiguity function, and the result in Proposition~\ref{prop:mainlobe_ellipse} still applies. For ambiguity function optimization, the mainlobe width $\mathrm{W}_{\mathrm{main}}$ should satisfy
$
	\mathrm{W}_{\mathrm{main}} \ge \frac{2\sqrt{2}}{Q}$.
This constraint prevents the mainlobe from becoming narrower than the cell diagonal, which would otherwise induce significant intra-cell sidelobe responses. Such responses may generate spurious peaks within a grid cell and thereby introduce estimation bias by shifting the estimated coordinates away from the true target locations.


\subsection{2D Posterior-Probability-Based Ambiguity Function}
In multi-target sensing, the sidelobes of the 2D ambiguity function characterize the interference that one target may impose on the estimation of another target, and are therefore closely related to sensing performance.  The relevant interference is determined by the angular  spatial-frequency offset pair $\{\Delta f_t,\Delta f_r\}$ between two targets. Although the 2D ambiguity function characterizes the array response over the entire 2D angular domain, it does not explicitly indicate which sidelobe or overlap regions are actually associated with the target separations in a given multi-target scene. To illustrate this point, Fig.~\ref{fig0331} shows a 1D slice of the 2D ambiguity function, where $\delta_f$ denotes the target  offset along the plotted dimension. The shaded region represents the interference from Target~1 to Target~2, showing that the relevant interference is determined by the actual target offset rather than the entire ambiguity surface.

\begin{figure}[t]
	\centering
	\includegraphics[width=2.3in]{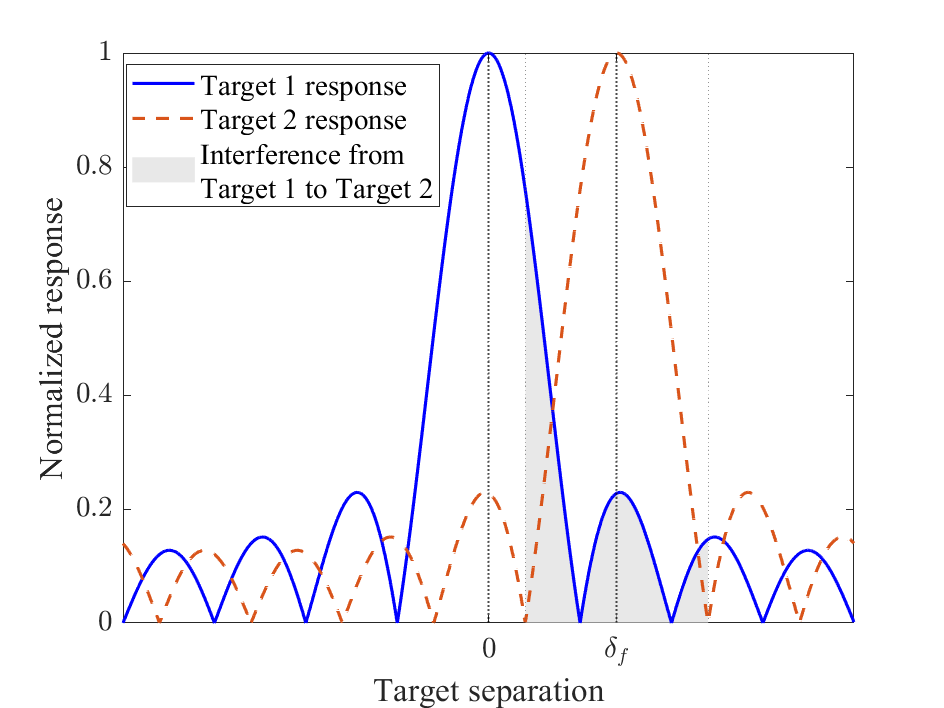}
	\caption{Illustration of inter-target interference in the ambiguity function.}
	\vspace{-3mm}
	\label{fig0331}
\end{figure}
Since only a finite number of targets are present, the estimation performance is mainly affected by the sidelobe values at the spatial-frequency offsets corresponding to the actual pairwise separations among the targets. Motivated by this observation, we propose a posterior-probability-based ambiguity function as a more informative metric for evaluating the overall multi-target sensing performance of  wireless sensing systems. Different from  the conventional ambiguity function,  the proposed ambiguity function exploits the coarse location information of multiple targets to infer their relative spatial frequency offsets $\{\Delta f_t,\Delta f_r\}$, thereby emphasizing the sidelobe regions that are more relevant to sensing performance degradation, as illustrated in Fig.  \ref{fig0331}. Specifically, the coarse target locations can be obtained from the E-step of the VBI algorithm in \cite{liu2020robust} under a ULA assumption. Therein, the transmit and receive angles are each discretized into $Q$ grid points over $[-\pi/2,\pi/2)$, forming a $Q\times Q$ joint angle grid, and the resulting posterior distribution $q(\mathbf{s})$ provides the probability that a target is present at each sensing grid point. To further map the coarse target-location information characterized by $q(\mathbf{s})$ into the spatial frequency offset domain, we partition the 2D ambiguity function along the transmit and receive spatial-frequency offset dimensions into $Q$ intervals, thereby forming a $Q\times Q$ offset grid. For each offset bin, its weight is obtained by traversing all candidate target pairs and aggregating their joint posterior probabilities, where only those pairs whose relative spatial-frequency offsets fall into that bin are counted. In this way, the probability associated with each offset bin characterizes the likelihood of the corresponding target separation and thus reflects the impact of the ambiguity sidelobes in that offset bin on multi-target sensing performance.

To illustrate the probability-mapping procedure, we take the $(k,f)$-th offset bin as an example, as shown in Fig.~\ref{fig3}. Suppose that two targets are located in the $(i,j)$-th and $(i',j')$-th sensing grid bins, respectively, and  their relative spatial-frequency offset can be mapped to the $(k,f)$-th offset bin. Then, according to the positions of these two sensing grid bins, the corresponding transmit and receive offset ranges are given by $\Delta f_t\in [\Delta f_{t,\mathrm{min}},\Delta f_{t,\mathrm{max}})$ and $\Delta f_r\in[\Delta f_{r,\mathrm{min}},\Delta f_{r,\mathrm{max}})$, where $\Delta f_{t,\mathrm{min}} \triangleq (i'-i-1)\tfrac{2}{Q}$, $\Delta f_{t,\mathrm{max}}\triangleq (i'-i+1)\tfrac{2}{Q}$, $\Delta f_{r,\mathrm{min}} \triangleq (j'-j-1)\tfrac{2}{Q}$, and $\Delta f_{r,\mathrm{max}}\triangleq(j'-j+1)\tfrac{2}{Q}$. Since the target locations lie in $[-1,1)$, the resulting offsets may extend over $(-2,2)$. Owing to the period-$2$ property of the ambiguity function, these offsets are interpreted modulo $2$ on the principal interval $[-1,1)$, which can be expressed as
\begin{equation}
	\Delta f_t\in
	\Big(\big[(i'-i-1)\tfrac{2}{Q},\, (i'-i+1)\tfrac{2}{Q}\big)\bmod 2\Big)
\end{equation}
and		
\begin{equation}
	\Delta f_r\in	\Big(\big[(j'-j-1)\tfrac{2}{Q},\, (j'-j+1)\tfrac{2}{Q}\big)\bmod 2\Big).
\end{equation}
 Hence, if the relative spatial-frequency offset of the two targets falls into the $(k,f)$-th offset bin, then the corresponding grid indices $k$ and $f$ should satisfy 
$
	k \!\in\!\!\bigg[(i'-i-1)\!\!\!\!\mod Q+\frac{Q}{2},(i'-i+1)\!\!\!\!\mod Q+\frac{Q}{2}\bigg)\!, k \!\in\mathbb{Z}$,
and
$
	f \!\in\!\!\bigg[(j'-j-1)\!\!\!\!\mod Q+\frac{Q}{2},(j'-j+1)\!\!\!\!\mod Q+\frac{Q}{2}\bigg),  f \!\in\mathbb{Z}$,
respectively. 
Equivalently,
$
k=(i'-i+\delta_1)\mod Q+c
$
and
$
f=(j'-j+\delta_2)\mod Q+c,
$
where $\delta_1,\delta_2\in\{-1,0\}$ and  $c\triangleq\lceil Q/2\rceil$. Conversely, for a given reference bin $(i,j)$ and offset bin $(k,f)$, the paired bin $(i',j')$ associated with this  offset can be recovered as
\begin{equation}\label{eq:0108_3_inv}
	i'
	=
	1+\Big(i-1+(k-c)-\delta_1\Big)\mod Q,
\end{equation}
and
\begin{equation}\label{eq:0108_4_inv}
	j'
	=
	1+\Big(j-1+(f-c)-\delta_2\Big)\mod Q.
\end{equation}
%
%
Then, we use $P_{k,f}=1$ to indicate that there exists at least one possible target pair whose relative spatial-frequency offset is mapped to the $(k,f)$-th offset bin. Accordingly, the activation probability associated with the $(k,f)$-th offset bin can be approximated as
\begin{equation}
	\label{eq:Pkf_simple2}
	p(P_{k,f}=1)
	=
	1-
	\prod_{i=1}^{Q}\prod_{j=1}^{Q}\prod_{\delta_1=-1}^{0}\prod_{\delta_2=-1}^{0}
	\Bigl(1-\eta_{i,j}\,\eta_{i',\,j'}\Bigr),
\end{equation}
where $\eta_{i,j}\triangleq q(s_{(i-1)Q+j})$. Based on the activation probabilities associated with all offset bins, the posterior-probability-based ambiguity function can be expressed as
\begin{equation}\label{eq:0414_2}
	\begin{aligned}
		&\chi_{{p}}(\Delta f_t,\Delta f_r)
		\!\!\triangleq
		\!\!\sum_{k=1}^{Q}\sum_{f=1}^{Q}
		p(P_{k,f})
		\chi(\Delta f_t,\Delta f_r)
		\mathbf 1_{\mathcal C_{k,f}}\!(\Delta f_t,\Delta f_r),
	\end{aligned}
\end{equation}
where
$
\mathcal C_{k,f}
\triangleq
\{(\Delta f_t,\Delta f_r)|
\Delta f_t\in I_{t,k},\;
\Delta f_r\in I_{r,f} \}$,
$
I_{t,k} \triangleq  [-1+(k-1)\tfrac{2}{Q},\, -1+k\tfrac{2}{Q} ),
\quad k\in\mathcal{Q}$,
$
I_{r,f} \triangleq  [-1+(f-1)\tfrac{2}{Q},\, -1+f\tfrac{2}{Q} ),
\quad f\in \mathcal{Q}$,
and $\mathbf 1_{\mathcal C_{k,f}}(\Delta f_t,\Delta f_r)$ denotes the indicator function of $\mathcal C_{k,f}$, i.e.,
\begin{equation}
	\mathbf 1_{\mathcal C_{k,f}}(\Delta f_t,\Delta f_r)\triangleq
	\begin{cases}
		1, & (\Delta f_t,\Delta f_r)\in \mathcal C_{k,f}, \\[2pt]
		0, & \text{otherwise}.
	\end{cases}
\end{equation}
\vspace{-3mm}


\begin{figure}[t]
	\centering
	\includegraphics[width=3in]{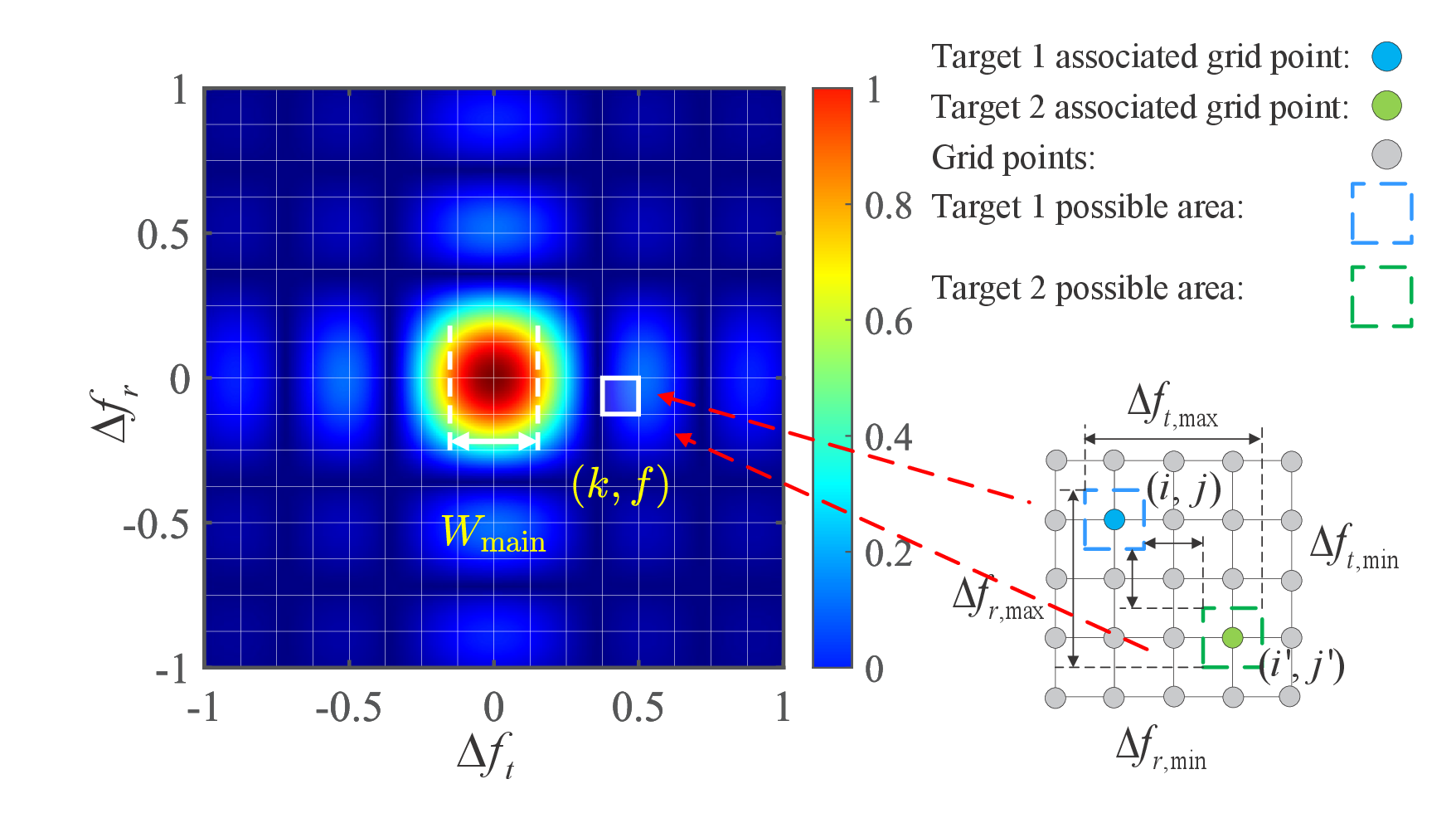}
	\vspace{-3mm}
	\caption{Illustration of  offset between  estimated targets.}
	\vspace{-3mm}
	\label{fig3}
\end{figure}

\subsection{Problem Formulation}
To improve the sensing performance in the considered multi-target setting, we propose to  optimize the MA positions based on $\chi_{{p}}(\Delta f_t,\Delta f_r)$, aiming to reduce the mainlobe width and suppress the sidelobe level while enforcing that the mainlobe width is no smaller than the grid resolution. The optimization problem is formulated as
\begin{equation}\label{eq:1217_1}
	\min_{(\mathbf{x}_t,\mathbf{x}_r)\in\mathcal{D}}
	J(\mathbf{x}_t,\mathbf{x}_r).
\end{equation}
where the objective function $J(\mathbf{x}_t,\mathbf{x}_r)$ is defined as
\begin{equation}
	\begin{aligned}
		J(\mathbf{x}_t,\mathbf{x}_r)
		\triangleq
		\int_{-1}^{1}\!\!\int_{-1}^{1}
		\sum_{k=1}^{Q}\sum_{f=1}^{Q}
		\bigl|\chi_{p}(\Delta f_t,\Delta f_r)\bigr|^2\,
		\mathrm{d}\Delta f_t\,\mathrm{d}\Delta f_r.
	\end{aligned}
\end{equation}
Here, the feasible set $\mathcal{D}$ is given by
$
\mathcal{D}= \mathcal{D}_{\mathrm{poly}}\cap \mathcal{D}_{\mathrm{quad}}$,
where $\mathcal{D}_{\mathrm{poly}}$ denotes the convex polytope induced by the linear constraints
	\begin{subequations}\label{eq:Dpoly}
		\begin{align}
			& 0 \le x_{t,1} \le \cdots \le x_{t,M_t} \le D_t, \label{eq:Dpoly_a}\\
			& x_{t,m+1}-x_{t,m} \ge x_{\min},\quad m=1,\ldots,M_t-1, \label{eq:Dpoly_b}\\
			& 0 \le x_{r,1} \le \cdots \le x_{r,M_r} \le D_r, \label{eq:Dpoly_c}\\
			& x_{r,n+1}-x_{r,n} \ge x_{\min},\quad n=1,\ldots,M_r-1, \label{eq:Dpoly_d}
		\end{align}
	\end{subequations}
where  $x_{\mathrm{min}}$ denotes the minimum allowable inter-element spacing,  $D_t$ and $D_r$ represent the maximum deployment spans of the transmit and receive antenna arrays, respectively, and  $\mathcal{D}_{\mathrm{quad}}$ denotes the feasible set associated with the nonlinear constraint $W_{\mathrm{main}}\ge \frac{2}{Q}\sqrt{2}$.
\section{ Dykstra-based Projected Gradient Descent  Algorithm}
The objective function of   problem  \eqref{eq:1217_1}  is highly nonconvex, which makes classical convex-optimization methods inapplicable. In addition, it contains multiple integrals, rendering the gradient evaluation cumbersome. Moreover, since the feasible set $\mathcal{D}$ is given by the intersection of two constraint sets, i.e.,$\mathcal{D}_{\mathrm{poly}}$ and $\mathcal{D}_{\mathrm{quad}}$, each gradient-descent update must be followed by a projection step to ensure feasibility, while direct projection onto $\mathcal{D}$ is of high computational complexity. To overcome these difficulties, we propose a low-complexity DPGD algorithm. Specifically, we first approximate the multiple-integral objective by a Riemann-sum form to facilitate gradient computation. Then, given that projections onto $\mathcal{D}_{\text{poly}}$ and $\mathcal{D}_{\text{quad}}$ can both be efficiently computed, we adopt Dykstra’s projection method [31] to find the projection onto $\mathcal{D} = \mathcal{D}_{\text{poly}} \cap \mathcal{D}_{\text{quad}}$ via alternating projections. The overall procedure alternates between a gradient-based update and the aforementioned Dykstra projection. This process is executed iteratively until the objective function converges to a stable solution.
\subsection{Gradient Computation of the Objective Function}
To evaluate the gradient of the objective function, we first approximate $J(\mathbf{x}_t,\mathbf{x}_r)$  by a Riemann-sum discretization, which can be written as
\begin{equation}\label{eq:J_riemann}
	\begin{aligned}
		&\bar{J}(\mathbf x_t,\mathbf x_r)
		\\&=
		h^2
		\sum_{k=1}^{Q}\sum_{f=1}^{Q}
		p(P_{k,f})
		\sum_{u=1}^{N}\sum_{v=1}^{N}
		\left|
		\chi_{{p}}\!\left(\Delta f_t^{(k,u)},\,\Delta f_r^{(f,v)}\right)
		\right|^2\!\!,
	\end{aligned}
\end{equation}
where the normalized offset domain $\Delta f_t\in[-1,1)$ and $\Delta f_r\in[-1,1)$ is first partitioned into $Q\times Q$ offset bins, and each bin is then further uniformly divided into $N\times N$ subcells for the Riemann-sum approximation. Accordingly, the sampling interval along each dimension is $h \triangleq \frac{2}{QN}$.
For the $(k,f)$-th offset bin, the sampled values along the transmit- and receive-offset dimensions are respectively given by
\begin{equation}
	\Delta f_t^{(k,u)} \triangleq -1+(k-1)\frac{2}{Q}+\Big(u-\tfrac{1}{2}\Big)h,\quad u=1,\ldots,N,
	\vspace{-1mm}
\end{equation}
\begin{equation}
	\Delta f_r^{(f,v)} \triangleq -1+(f-1)\frac{2}{Q}+\Big(v-\tfrac{1}{2}\Big)h,\quad v=1,\ldots,N,
	\vspace{-1mm}
\end{equation}
which correspond to the midpoints of the $N$ uniform subintervals along the transmit- and receive-offset dimensions within that bin.

Then, we update $\mathbf{x}_t$ and $\mathbf{x}_r$ by taking a descent step along their respective negative gradients, i.e.,
\begin{equation}\label{eq:0120_1}
	\mathbf{x}_t^{k+1}=\mathbf{x}_t^{k}-\omega_t^{k}\nabla_{\mathbf{x}_t}\bar{J}\!\left(\mathbf{x}_t^{k},\mathbf{x}_r^{k}\right),
\end{equation}
and
\begin{equation}\label{eq:0120_2}
	\mathbf{x}_r^{k+1}=\mathbf{x}_r^{k}-\omega_r^{k}\nabla_{\mathbf{x}_r}\bar{J}\!\left(\mathbf{x}_t^{k},\mathbf{x}_r^{k}\right),
\end{equation}
where $k$ denotes the iteration index, and $\omega_t^{k}$ and $\omega_r^{k}$ are the step sizes determined by the Armijo backtracking line search. $\nabla_{\mathbf{x}_t}\bar{J}(\mathbf{x}_t,\mathbf{x}_r)$ and $\nabla_{\mathbf{x}_r}\bar{J}(\mathbf{x}_t,\mathbf{x}_r)$ denote the gradients of $\bar{J}(\mathbf{x}_t,\mathbf{x}_r)$ with respect to $\mathbf{x}_t$ and $\mathbf{x}_r$, respectively, given by
\begin{equation}
	\nabla_{\mathbf{x}_t}\bar{J}(\mathbf{x}_t,\mathbf{x}_r)
	=
	\Big[\tfrac{\partial \bar{J}(\mathbf{x}_t,\mathbf{x}_r)}{\partial x_{t,1}},
	\tfrac{\partial \bar{J}(\mathbf{x}_t,\mathbf{x}_r)}{\partial x_{t,2}},
	\ldots,
	\tfrac{\partial \bar{J}(\mathbf{x}_t,\mathbf{x}_r)}{\partial x_{t,M_t}}\Big]^{\mathrm T}
\end{equation}
and
\begin{equation}
	\nabla_{\mathbf{x}_r}\bar{J}(\mathbf{x}_t,\mathbf{x}_r)
	=
	\Big[\tfrac{\partial \bar{J}(\mathbf{x}_t,\mathbf{x}_r)}{\partial x_{r,1}},
	\tfrac{\partial \bar{J}(\mathbf{x}_t,\mathbf{x}_r)}{\partial x_{r,2}},
	\ldots,
	\tfrac{\partial \bar{J}(\mathbf{x}_t,\mathbf{x}_r)}{\partial x_{r,M_r}}\Big]^{\mathrm T},
\end{equation}
where
\begin{equation}\label{eq:grad_xt_final}
	\begin{aligned}
		&		\frac{\partial 	\bar{J}(\mathbf x_t,\mathbf x_r)}{\partial x_{t,m}}
		=
		\frac{2h^2}{M_t^2M_r^2}
		\sum_{k=1}^{Q}\sum_{f=1}^{Q} p(P_{k,f})\!
		\sum_{u=1}^{N}\sum_{v=1}^{N}
		\bigl|S_r(\Delta f_r^{(f,v)})\bigr|^2\,
		\\&\Re\!\left\{
		S_t^*(\Delta f_t^{(k,u)})\,(j\frac{2\pi}{\lambda}\,\Delta f_t^{(k,u)})\,
		e^{j\frac{2\pi}{\lambda} x_{t,m}\Delta f_t^{(k,u)}}
		\right\},
	\end{aligned}
\end{equation}
and
\begin{equation}\label{eq:grad_xr_final}
	\begin{aligned}
		&		\frac{\partial 	\bar{J}(\mathbf x_t,\mathbf x_r)}{\partial x_{r,n}}
		=
		\frac{2h^2}{M_t^2M_r^2}
		\sum_{k=1}^{Q}\sum_{f=1}^{Q} p(P_{k,f})\!\!
		\sum_{u=1}^{N}\sum_{v=1}^{N}
		\bigl|S_t(\Delta f_t^{(k,u)})\bigr|^2\,
		\\&\Re\!\left\{
		S_r^*(\Delta f_r^{(f,v)})\,(j\frac{2\pi}{\lambda}\,\Delta f_r^{(f,v)})\,
		e^{j\frac{2\pi}{\lambda} x_{r,n}\Delta f_r^{(f,v)}}
		\right\}.
	\end{aligned}
\end{equation}
Here,
$
S_t(\Delta f_t)\triangleq \sum_{m=1}^{M_t} e^{j\frac{2\pi}{\lambda} x_{t,m}\Delta f_t}
$
and
$
S_r(\Delta f_r)\triangleq \sum_{n=1}^{M_r} e^{j\frac{2\pi}{\lambda} x_{r,n}\Delta f_r}
$
are adopted for notational simplicity.

\subsection{Dykstra Projection Method}

\begin{figure}[t]
	\centering
	\includegraphics[width=3in]{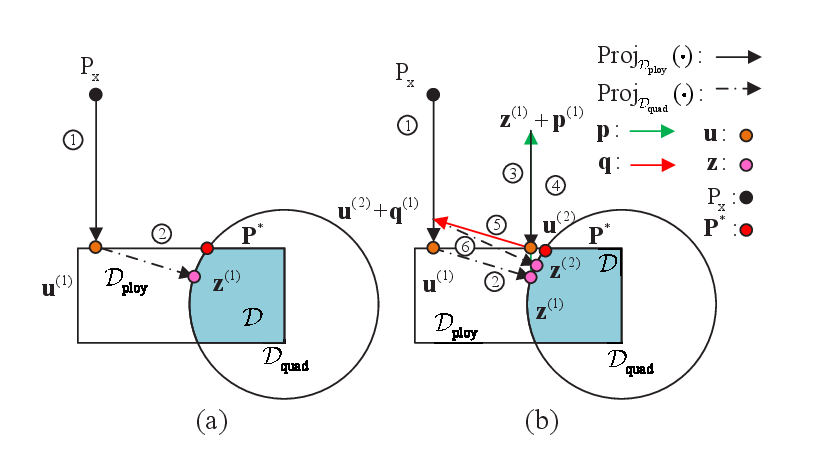}
	\vspace{-3mm}
	\caption{Illustration of (a) POCS and (b) Dykstra projection method.}
	\vspace{-3mm}
	\label{fig:0126_1}
\end{figure}

To guarantee that the iterate after the gradient-descent update remains feasible, we formulate the following projection problem:
\begin{equation}\label{eq:0107_1}
	\operatorname{Proj}_{\mathcal D}({\mathbf x})
	=
	\arg\min_{\tilde{\mathbf x}\in\mathcal D}\ \frac{1}{2}\|\tilde{\mathbf x}-{\mathbf x}\|_2^2,
\end{equation}
where $\mathbf{x}\triangleq\big[\mathbf{x}_t^{ T},\,\mathbf{x}_r^{ T}\big]^{ T}$. Since $\mathcal{D}=\mathcal{D}_{\mathrm{poly}}\cap\mathcal{D}_{\mathrm{quad}}$, and the projections onto $\mathcal{D}_{\mathrm{poly}}$ and $\mathcal{D}_{\mathrm{quad}}$ can both be computed with low complexity, it is natural to construct the projection onto $\mathcal{D}$ based on alternating projections onto $\mathcal{D}_{\mathrm{poly}}$ and $\mathcal{D}_{\mathrm{quad}}$. We first consider the projection onto convex sets (POCS) method \cite{650118}, which can find a feasible point in $\mathcal{D}$ via successive projections onto the individual sets, as illustrated in Fig.~\ref{fig:0126_1} (a), where $\mathbf{P}_{\mathbf{x}}$ denotes a point outside the feasible set $\mathcal{D}$, $\mathbf{P}^{*}$ denotes its exact Euclidean projection onto $\mathcal{D}$, and $\mathbf{u}$ and $\mathbf{z}$ denote the projection points onto $\mathcal{D}_{\mathrm{poly}}$ and $\mathcal{D}_{\mathrm{quad}}$, respectively. However, as can be seen, POCS does not generally converge to $\mathbf{P}^{*}$. To overcome this issue, we adopt Dykstra projection method \cite{boyle1986method}, which incorporates correction terms into alternating projections and guarantees convergence to $\mathbf{P}^{*}$, as illustrated in Fig.~\ref{fig:0126_1} (b). Here, $\mathbf p$ and $\mathbf q$ denote the correction variables added before the projections onto $\mathcal D_{\mathrm{poly}}$ and $\mathcal D_{\mathrm{quad}}$, respectively. Accordingly, for problem \eqref{eq:0107_1}, the correction variables are updated at iteration $l{+}1$ as
\begin{equation}\label{eq:0302_3}
	\mathbf p^{(l+1)}=\mathbf z^{(l)}+\mathbf p^{(l)}-\mathbf u^{(l+1)},
\end{equation}
\begin{equation}\label{eq:0302_4}
	\mathbf q^{(l+1)}=\mathbf u^{(l+1)}+\mathbf q^{(l)}-\mathbf z^{(l+1)},
\end{equation}
respectively. The algorithm is initialized with $\mathbf z^{(1)}=\mathbf x$, $\mathbf p^{(1)}=\mathbf 0$, and $\mathbf q^{(1)}=\mathbf 0$. Then, the iterative update rules for the alternating-projection procedure can be expressed as
\begin{equation}\label{eq:0302_1}
	\mathbf u^{(l+1)}=\operatorname{Proj}_{\mathcal D_\mathrm{ploy}}\!\big(\mathbf z^{(l)}+\mathbf p^{(l)}\big),
\end{equation}
\begin{equation}\label{eq:0302_2}
	\mathbf z^{(l+1)}=\operatorname{Proj}_{\mathcal D_\mathrm{quad}}\!\big(\mathbf u^{(l+1)}+\mathbf q^{(l)}\big).
\end{equation}
Since the transmit and receive antennas are subject to identical constraints, we take the transmit antenna as an example to illustrate the projections  onto $\mathcal D_{\mathrm{poly}}$ and $\mathcal D_{\mathrm{quad}}$. 

\textit{1) $\operatorname{Proj}_{\mathcal D_\mathrm{ploy}}(\cdot)$}: 
Since directly projecting onto $\mathcal{D}_{\mathrm{poly}}$ is complicated by the additional minimum-spacing constraint \eqref{eq:Dpoly_a}, we transform it into an equivalent standard monotonicity constraint to enable low-complexity projection. Specifically, let $\mathbf{k}_t \triangleq [k_{t,1},k_{t,2},\cdots,k_{t,M_t}]^T$, where 
$
k_{t,m}\triangleq x_{t,m}-(m-1)x_{\mathrm{min}},m=1,\cdots,M_t$,
and $\mathbf{j}_t\triangleq[0,x_{\mathrm{min}},2x_{\mathrm{min}},\cdots,(M_t-1)x_{\mathrm{min}}]^T$. Then, we can obtain  $\mathbf{x}_t=\mathbf{k}_t+\mathbf{j}_t$, and the  constraints \eqref{eq:Dpoly_a} and  \eqref{eq:Dpoly_b} can be  equivalently rewritten as the following monotone constraint: 
$
0\leq k_{t,1}\leq k_{t,2}\leq\cdots\leq k_{t,M_t}\leq D'_t$,
where $D'_t\triangleq D_t-(M_t-1)x_{\mathrm{min}}$. Then, the projection of $\mathbf{x}_t$ onto  $\mathcal{D}_{\mathrm{poly}}$ can be formulated as
\begin{equation}\label{eq:0127_2}
	\begin{aligned}
		&\operatorname{Proj}_{\mathcal D_\mathrm{ploy}}(\mathbf{d}_t)
		\\&= \mathbf{j}_t + \arg\min_{0 \le k_{t,1} \le \cdots \le k_{t,M_t} \le D'_t}
		\frac{1}{2}\left\lVert \mathbf{k}_t - (\mathbf{x}_t - \mathbf{j}_t)\right\rVert_2^2,
	\end{aligned}
\end{equation}
which is a box-constrained isotonic regression problem, that can be efficiently solved by the pool-adjacent-violators (PAV) algorithm  with $\mathcal{O}(M_t)$ complexity \cite{chen2025covariancesurrogateframeworkmovableantennaenabled}. The details are omitted for brevity.

\textit{2) $\operatorname{Proj}_{\mathcal D_\mathrm{quad}}(\cdot)$}:
 As indicated by \eqref{eq:0409_1}, the quadratic constraint is governed by the spatial variance of the antenna positions, characterizing how these elements are distributed relative to their centroid $\hat{x}_t$. Thus, this constraint controls the spread of the array relative to its center, instead of its absolute position. Accordingly, the Euclidean projection onto $\mathcal D_{\mathrm{quad}}$ can be obtained by fixing $\hat{x}_t$ and scaling all antenna positions toward this center by a common factor until the constraint is met. In particular, if $\sigma_t^2 \le {Q^2(1-\rho_6)}/{2}$, then $\mathbf{x}_t$ is already feasible; otherwise, when $\sigma_t^2 > {Q^2(1-\rho_6)}/{2}$, the projection is given by
\begin{equation}
	\operatorname{Proj}_{\mathcal D_{\mathrm{quad}}}(\mathbf{x}_t)
	= \hat x_t \mathbf{1} + \nu\left(\mathbf{x}_t - \hat x_t \mathbf{1}\right)
	= \hat x_t \mathbf{1} + \nu\,{\mathbf{\bar{x}}}_t,
\end{equation}
where 
$
	\nu
	= 
	\frac{\lambda Q}{2\pi}\sqrt{\frac{(1-\rho_6)M_t}{2}}\,
	\frac{1}{\left\lVert \mathbf{\bar{x}}_t\right\rVert_2}$.

Algorithm 1 summarizes the proposed DPGD algorithm. The proposed algorithm is essentially a standard PGD method over the feasible set $\mathcal{D}$, i.e.,
$
\mathbf{x}^{k+1}
=
\operatorname{Proj}_{\mathcal D}\!\left(\mathbf{x}^{k}-[\omega^k_t \nabla_{\mathbf{x}_t} \bar{J}\!\left(\mathbf{x}_t^{k},\mathbf{x}_r^{k}\right);\omega^k_r \nabla_{\mathbf{x}_r} \bar{J}\!\left(\mathbf{x}_t^{k},\mathbf{x}_r^{k}\right)]\right)$. Its convergence follows from standard results on PGD under the adopted assumptions on the objective function and the stepsize rule \cite{1101194}. In our implementation, the projection onto $\mathcal{D}$ is computed by the Dykstra projection method, since $\mathcal{D}$ is the intersection of closed convex sets. As an inner routine for evaluating $\operatorname{Proj}_{\mathcal D}(\cdot)$, the Dykstra projection method does not alter the outer PGD iteration. Moreover, the Dykstra projection method converges to the exact Euclidean projection onto the intersection of finitely many closed convex sets \cite{boyle1986method}. Therefore, when the inner Dykstra iteration is solved to sufficient accuracy at each step, the convergence conclusion of the proposed algorithm remains the same as that of standard PGD.

 As for its computational complexity,  the main burden arises from calculating the gradients $	\nabla_{\mathbf{x}_t}\bar{J}(\mathbf{x}_t,\mathbf{x}_r)$ and  $\nabla_{\mathbf{x}_r}\bar{J}(\mathbf{x}_t,\mathbf{x}_r)$, whose complexity are on the orders of $\mathcal{O}(M_t Q^2N^2)$ and $\mathcal{O}(M_r Q^2N^2)$, respectively. Moreover,  the projection steps onto $\mathcal{D}_{\mathrm{poly}}$ and $\mathcal{D}_{\mathrm{quad}}$ both have computational complexity  of $\mathcal{O}(M_t+M_r)$. Therefore, the overall computational complexity of Algorithm~1 can be expressed as
$\mathcal{O}\big(K_1(M_t+M_r)Q^2N^2 + K_1K_2(M_t+M_r)\big)$,
where $K_1$ denotes the number of gradient-descent iterations and $K_2$ denotes the number of alternating projection steps performed per gradient-descent iteration. In contrast, the complexity of the conventional PGD algorithm (e.g., an interior point method (IPM)-based PGD algorithm \cite{8268092}) for problem \eqref{eq:1217_1} is $\mathcal{O}(K_1(M_t+M_r)Q^2N^2 + K_1(M_t+M_r)^3)$. Thus, the proposed DPGD algorithm exhibits  a significantly lower computational complexity than the conventional approach.

\begin{algorithm}[t]
	\caption{Proposed Dykstra-based Projected Gradient Descent  Algorithm}
	\label{alg:1}
	\renewcommand{\algorithmicrequire}{\textbf{Input:}}
	\renewcommand{\algorithmicensure}{\textbf{Output:}}
	\begin{algorithmic}[1]
		\REQUIRE   Outer maximum iterations $K$, inner maximum iterations $L$, threshold $T$, step sizes $\omega_{t}$ and $\omega_{r}$, and initial points $\mathbf{x}_{t}^{0}$ and $\mathbf{x}_{r}^{0}$.
		\ENSURE Optimized antenna distribution  $\mathbf x^\star_t$ and $\mathbf x^\star_r$.
		\STATE Set $k\!\leftarrow\!0$.
		\WHILE{$k< K$}
		\STATE $k\leftarrow k+1$.
		\STATE Update $\mathbf{x}^{(k)}_t$ and $\mathbf{x}^{(k)}_r$, using \eqref{eq:0120_1} and \eqref{eq:0120_2}.
		\STATE $l\!\leftarrow\!0$.
		\STATE \textbf{Initialize:} $\mathbf{z}^{(0)} = [\mathbf{x}^{(k)T}_t,\mathbf{x}^{(k)T}_r]^T$, $\mathbf{u}^{(0)} = \mathbf{0}$, $\mathbf{p}^{(0)} = \mathbf{0}$ and $\mathbf{q}^{(0)} = \mathbf{0}$.
		\WHILE{$l<L$}
		\STATE $l\leftarrow l+1$.
		\STATE Update $\mathbf{u}^{(l)}$ and  $\mathbf{z}^{(l)}$, using \eqref{eq:0302_1} and \eqref{eq:0302_2}.
		\STATE Update $\mathbf{p}^{(l)}$ and $\mathbf{q}^{(l)}$, using \eqref{eq:0302_3} and \eqref{eq:0302_4}.
		\ENDWHILE
		\STATE  Extract $\mathbf{x}^{(k)}_t$ and $\mathbf{x}^{(k)}_r$ from $\mathbf{z}^{(l)}$.
		\IF{$\|\mathbf{x}^{(k)}_t-\mathbf{x}^{(k-1)}_t\|+\|\mathbf{x}^{(k)}_r-\mathbf{x}^{(k-1)}_r\|<T$}
		\STATE \textbf{break;}
		\ENDIF
		\ENDWHILE
		\RETURN $\mathbf x^\star_t=\mathbf{x}^{(k)}_t$ and $\mathbf x^\star_r=\mathbf{x}^{(k)}_r$.
	\end{algorithmic}
\end{algorithm}

\section{Numerical Results}
In this section, we present numerical results to validate the effectiveness of the proposed cross sparsity Markov mixture prior model, the effectiveness of the proposed ambiguity function for MA optimization, the estimation-accuracy gains enabled by MAs, and the efficacy of the proposed DPGD algorithm.
We consider a colocated MIMO {wireless sensing system} with $M_t=8$ transmit and $M_r=8$ receive MAs operating at a carrier frequency of $4\,\mathrm{GHz}$. The initial inter-element spacing of both arrays is set to $\lambda/2$, and the transmit- and receive-aperture sizes are set to $D_t=8\lambda$ and $D_r=8\lambda$, respectively. The angular domain is discretized into $Q=16$ grids.
Based on \cite{chen2024joint}, the large-scale fading models for the direct LoS paths,
$g_k$, $k \in \mathcal{K}$, and for the first-order NLoS paths,
$h_{i,j}$, $i,j \in \mathcal{K},\, i \neq j$, are expressed as
$
	g_k=\sqrt{\frac{\lambda^2 \kappa_k}{64\pi^3 d_k^4}}$
and
$
	h_{i,j}=\sqrt{\frac{\lambda^2\,\Psi_{i,j}\,\Psi_{j,i}}{(4\pi)^4\, d_i^2\, l_{i,j}^2\, d_j^2}}$,
respectively.
Here, $\kappa_k$ denotes the radar cross section (RCS) of the $k$-th target along the direct LoS path and is set to $-10\,\text{dBsm}$, while $\Psi_{i,j}$ denotes the bistatic RCS of the $i$-th target along the first-order NLoS path in the look  of the $j$-th target (i.e., scattering from the $i$-th target toward the $j$-th target) and is set to $0\,\text{dBsm}$ \cite{6374215}.  $d_k$ and $l_{i,j}$ denote the distance between the $k$-th target and the radar, and the distance between the $i$-th and $j$-th targets, respectively.
 The transmit power of the base station is set as $P_T=30\,\text{dBm}$.  We assume  that
the direct angles of the targets are randomly distributed within the interval $[-\frac{\pi}{2}, +\frac{\pi}{2}]$, with the constraint that they lie on distinct angular grid points.
$J=500$ independent Monte Carlo trails  are performed, and the estimation accuracy is evaluated in terms of the RMSE, defined as
$
	\mathrm{RMSE} = \sqrt{\frac {1}{K}\frac{1}{J} \sum_{k=1}^{K} 
		 \sum_{j=1}^{J} 
		\left( \theta_k^{(j)} - \hat{\theta}_k^{(j)} \right)^2 }$,
where $\theta_k^{(j)}$ and $\hat{\theta}_k^{(j)}$ denote the true and estimated direct LoS path angles of the $k$-th target in the $j$-th trial, respectively.

\begin{figure}[t]
	\centering
	\includegraphics[width=2in]{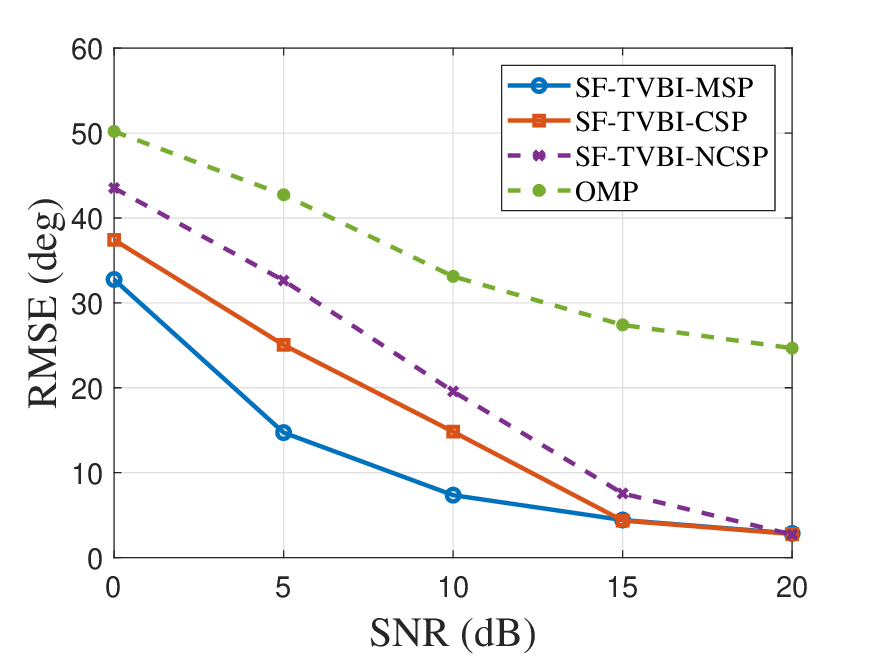}
	\vspace{-3mm}
	\caption{RMSE versus SNR for different prior models.}
	\vspace{-4mm}
	\label{fig0122_2}
\end{figure}
\begin{figure}[t]
	\centering
	\includegraphics[width=2in]{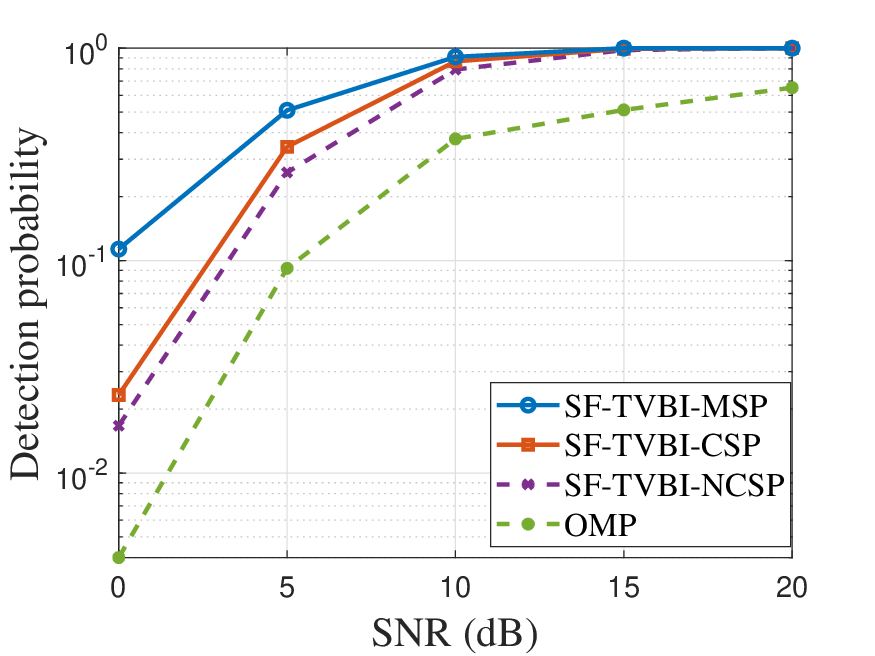}
		\vspace{-3mm}
	\caption{Detection probability versus SNR for different prior models.}
		\vspace{-5mm}
	\label{fig0122_3}
\end{figure}
\begin{figure}[t]
	\centering
	\includegraphics[width=2.3in]{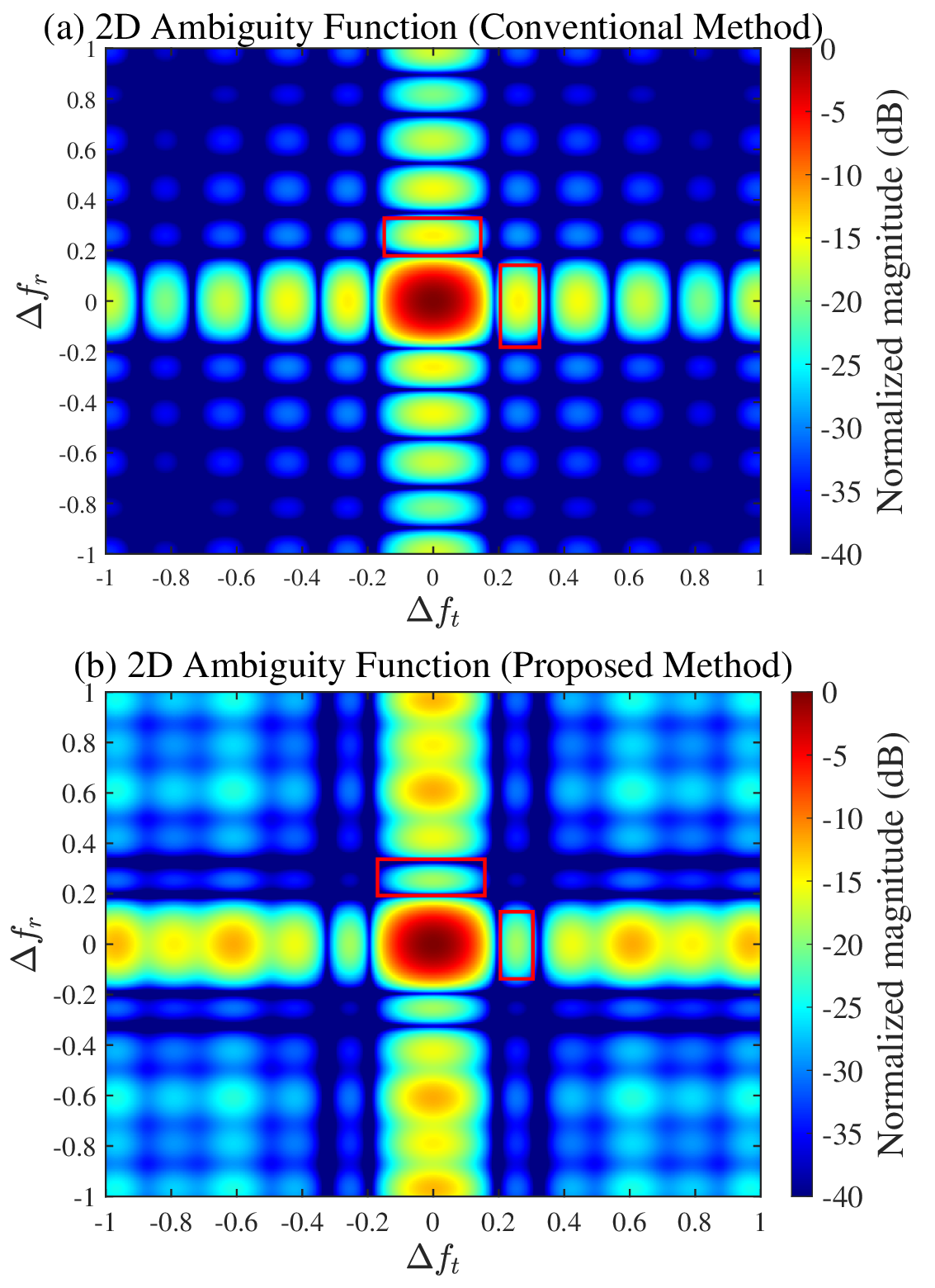}
		\vspace{-3mm}
	\caption{Comparison of 2D ambiguity functions for the traditional and proposed optimization methods.}
		\vspace{-4mm}
	\label{fig0122_8}
\end{figure}
\begin{figure}[t]
	\centering
	\includegraphics[width=2in]{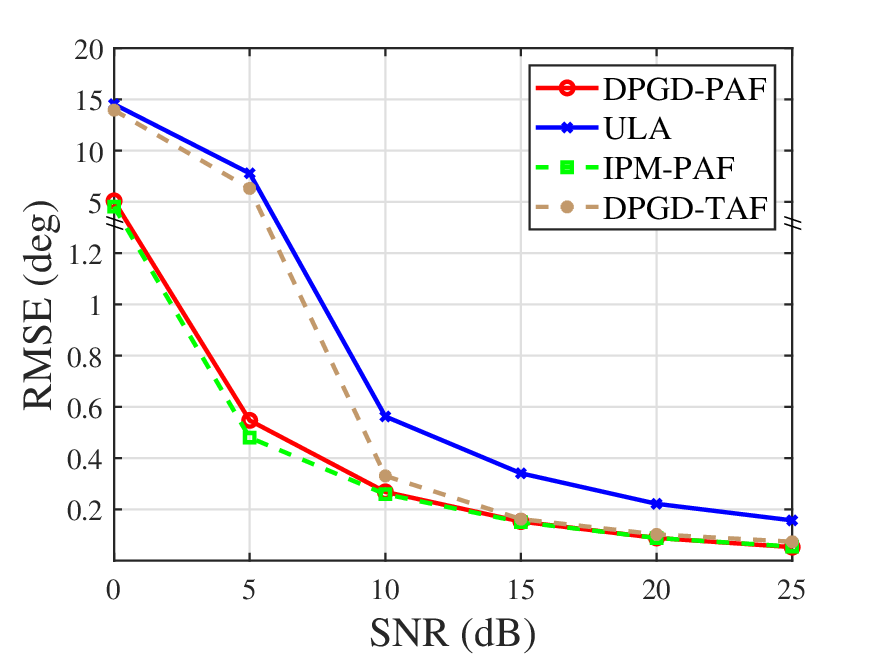}
	\vspace{-3mm}
	\caption{RMSE versus SNR for the MA-optimized antenna deployment.}
			\vspace{-5mm}
	\label{fig0122_5}
\end{figure}
\begin{figure}[t]
	\centering
	\includegraphics[width=2in]{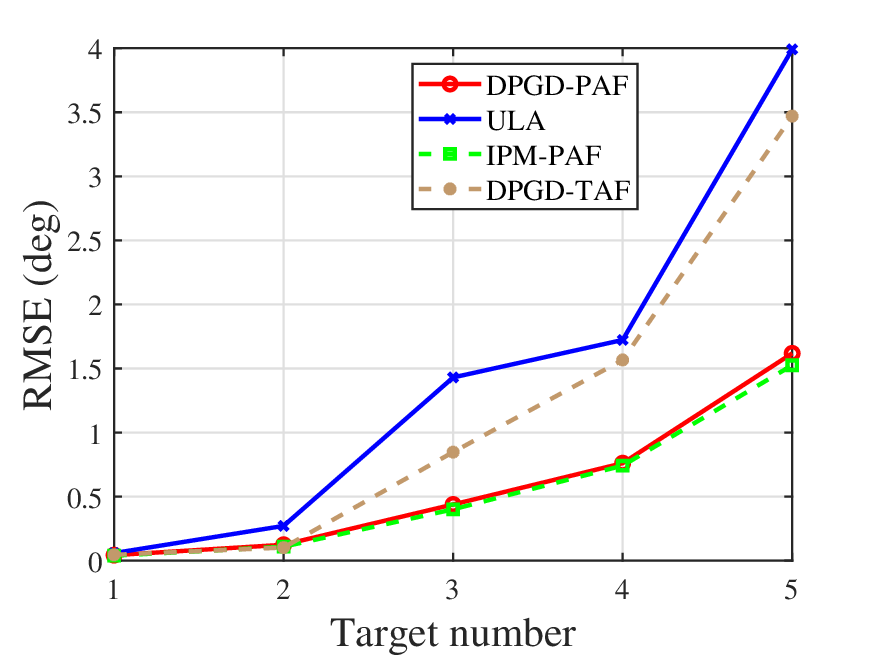}
		\vspace{-3mm}
	\caption{RMSE versus target number for the MA-optimized antenna deployment.}
		\vspace{-5mm}
	\label{fig0122_6}
\end{figure}


First, in Fig.~\ref{fig0122_2}, we evaluate the sensing performance gain of SF-TVBI with the proposed Markov mixture prior (SF-TVBI-MSP), where four targets are considered and their angles are randomly drawn from the interval $[-\pi/2,\pi/2]$. For comparison, we consider SF-TVBI without the cross sparsity prior model (SF-TVBI-NCSP), SF-TVBI with the cross sparsity prior model (SF-TVBI-CSP), and orthogonal matching pursuit (OMP) as benchmark methods. The results show that the SF-TVBI-based methods outperform OMP by at least $5\,\text{dB}$, mainly due to the use of probabilistic prior models to improve estimation accuracy. Moreover, SF-TVBI-CSP achieves about $3\,\text{dB}$ gain over SF-TVBI-NCSP, which demonstrates the benefit of exploiting the cross sparsity structure. Compared with SF-TVBI-CSP, the proposed SF-TVBI-MSP  provides around $5\,\text{dB}$ gain in the low-to-moderate SNR regime, since the proposed mixture prior can effectively alleviate the ambiguity caused by the cross sparsity structure. At very low SNRs, its performance becomes close to that of SF-TVBI-CSP, because the sparse structure is no longer sufficient to suppress noise-induced estimation errors. At high SNRs, the performance gap between the two methods becomes negligible, as the ambiguity caused by cross sparsity is largely removed and the advantage of the mixture prior correspondingly diminishes.

Next, in Fig. \ref{fig0122_3}, we evaluate the target detection probability of  SF-TVBI-MSP across a range of SNR levels. This metric is defined as the probability of targets being identified in the correct grid bin. As can be seen, the VBI-based algorithm can achieve better detection probability than the OMP algorithm. Moreover, SF-TVBI-CSP achieves better detection performance than SF-TVBI-NCSP, whereas the improvement is modest. This is mainly because, as the SNR decreases, the inter-grid ambiguity induced by cross sparsity reduces the probability of correct detection. In contrast, SF-TVBI-MSP offers an additional gain of about $3\,\text{dB}$ in terms of detection probability over SF-TVBI-CSP in the low-SNR regime. This gain arises because the proposed prior preserves the benefit of cross sparsity in enhancing detection via first-order NLoS paths, while suppressing the ambiguity introduced by cross sparsity through the Markov dependence.

In Fig.~\ref{fig0122_8}, we compare the 2D ambiguity functions obtained by the conventional ambiguity function optimization method \cite{san2007mimo} and the proposed method, where the two target locations are randomly generated, and the estimated target offsets are $\Delta f_t=0.28$ and $\Delta f_r=0.28$. It can be observed that the proposed method achieves a mainlobe width comparable to that of the traditional  method. More importantly, around the target offset location, i.e., near $\Delta f_t=0.28$ and $\Delta f_r=0.28$ as highlighted by the red boxes, the sidelobe level of the proposed method is significantly lower than that of the traditional  method. This demonstrates that the proposed method can effectively suppress the sidelobes at desired angular offsets while maintaining a narrow mainlobe width, thereby improving the sensing performance.

In Fig.~\ref{fig0122_5}, we evaluate the sensing performance under different SNRs to quantify the gain achieved by the proposed MA design, where the angles of the two targets are randomly drawn from the interval $[-\pi/2,\pi/2]$ are considered. For comparison, we consider DPGD with the proposed ambiguity function \eqref{eq:0414_2} (DPGD-PAF), DPGD with the traditional ambiguity function \eqref{eq:1130} (DPGD-TAF), IPM with the proposed ambiguity function (IPM-PAF), and a ULA-based antenna deployment. The results show that DPGD-PAF outperforms DPGD-TAF over the whole SNR range, which confirms that the proposed ambiguity function is more suitable  for antenna position optimization than the traditional one. This is because the proposed ambiguity function incorporates the coarse location information of multiple targets and places greater emphasis on the sidelobe regions associated with their actual relative separations, thereby producing an antenna deployment that is specifically optimized to mitigate inter-target interference and thus bolstering the overall efficacy of multi-target sensing. Furthermore, DPGD-PAF achieves nearly the same performance as IPM-PAF, which demonstrates that the proposed DPGD algorithm can provide competitive sensing performance with substantially lower computational complexity. Finally, both MA-optimized deployments achieve clear gains over the ULA benchmark, showing the advantage of adaptive antenna deployment in reducing estimation error.

Finally, in  Fig.~\ref{fig0122_6}, we investigate the RMSE performance of the considered methods under different number of targets, where the receive SNR is fixed at $\mathrm{SNR}=20\,\text{dB}$. It can be observed that, compared with DPGD-TAF, DPGD-PAF yields improved RMSE performance as the target number increases. This is because, when more targets are present, the sidelobe-induced inter-target interference becomes more severe, and the proposed ambiguity function can then better exploit the coarse location information of multiple targets to suppress the sidelobe interference associated with their separations, thereby leading to a more suitable antenna deployment for dense multi-target sensing. Moreover, it can be observed that, compared with the ULA benchmark, the performance gain provided by MAs becomes more pronounced as the target number increases. This is because MAs can flexibly adapt the antenna positions, and are therefore more capable of suppressing the increasingly severe sidelobe-induced inter-target interference in dense multi-target sensing scenarios.

\section{Conclusion}
In this paper, we investigated an MA-enabled  wireless sensing system for multi-target estimation under multipath propagation.  To improve sensing performance, we proposed a cross sparsity Markov mixture prior model to enhance direct LoS path estimation accuracy by exploiting the structural correlation between the LoS and NLoS paths. Furthermore, due to the difficulty of characterizing multi-target sensing performance in the presence of mutual target interference, we introduced a 2D ambiguity function and established its explicit relationship with the antenna spatial distribution. Based on this analysis, we developed a posterior-probability-based ambiguity function using the posterior probabilities of target locations obtained from coarse estimation, and formulated an MA distribution optimization problem to suppress sidelobes and narrow the mainlobe. To efficiently solve the resulting highly non-convex problem, we developed a low-complexity DGPD algorithm. Numerical results verified the sensing performance gains enabled by the proposed prior model and MAs, and demonstrated the efficiency of the proposed algorithm.
\vspace{-2mm}
{\appendix
\subsection{Proof of Proposition 1}
As can be seen from \eqref{eq:1131}, the considered 2D ambiguity function is constructed as the product of two terms, each of which is a normalized summation of complex exponential functions. Since the mainlobe region corresponds to sufficiently small frequency offsets $\{\Delta f_t,\Delta f_r\}$ around the peak, we characterize the local mainlobe behavior by applying Taylor expansion to the exponential terms in \eqref{eq:1131}. Specifically, 
$
	\frac{1}{M_t}\sum_{m=1}^{M_t} e^{j \alpha_{t,m}\Delta f_t}$ and $
	\frac{1}{M_r}\sum_{n=1}^{M_r} e^{j \alpha_{r,n}\Delta f_r}$
can be approximated up to the second order as
\begin{equation}
	\frac{1}{M_t}\sum_{m=1}^{M_t}
	e^{j \alpha_{t,m}\Delta f_t}
	\approx
	1 + j\mu_{1,t}\Delta f_t - \frac{\mu_{2,t}}{2}\Delta f_t^2
	\vspace{-2mm}
\end{equation}
and
\begin{equation}
	\frac{1}{M_r}\sum_{n=1}^{M_r}
	e^{j \alpha_{r,n}\Delta f_r}
	\approx
	1 + j\mu_{1,r}\Delta f_r - \frac{\mu_{2,r}}{2}\Delta f_r^2,
		\vspace{-2mm}
\end{equation}
respectively, where $\mu_{1,t}$, $\mu_{2,t}$, $\mu_{1,r}$, and $\mu_{2,r}$ are defined as in the main text.

By taking the squared magnitudes of the above second-order expansions and neglecting the resulting fourth-order terms $\Delta f_t^4$ and $\Delta f_r^4$, which are negligible in the local mainlobe region, we obtain
\begin{equation}\label{eq:0113_1}
	\left|\frac{1}{M_t}\sum_{m=1}^{M_t}
	e^{j \alpha_{t,m}\Delta f_t}\right|^2
	\approx 1 - \sigma_t^2\,\Delta f_t^2
	\vspace{-2mm}
\end{equation}
and
\begin{equation}\label{eq:0113_2}
	\left|\frac{1}{M_r}\sum_{n=1}^{M_r}
	e^{j \alpha_{r,n}\Delta f_r}\right|^2
	\approx 1 - \sigma_r^2\,\Delta f_r^2,
\end{equation}
where $\sigma_t^2$ and $\sigma_r^2$ are defined as in the main text.
Substituting \eqref{eq:0113_1} and \eqref{eq:0113_2} into \eqref{eq:1131} yields 
$
		|\chi(f_t,f_r,f'_t,f'_r)|^2
		\approx
		1-\sigma_t^2\Delta f_t^2-\sigma_r^2\Delta f_r^2$,
where the fourth-order coupling term $\sigma_t^2\sigma_r^2\Delta f_t^2\Delta f_r^2$ is omitted under the local approximation.

According to the adopted definition of the mainlobe region, the mainlobe boundary is defined as the set of points where $|\chi(f_t,f_r,f'_t,f'_r)|^2$ decreases by $6\,\mathrm{dB}$ from its peak value. Therefore, the boundary is characterized by
\begin{equation}\label{eq:0113_5}
	\sigma_t^2\Delta f_t^2 + \sigma_r^2\Delta f_r^2 = 1-\rho_6,
\end{equation}
where $\rho_6 \triangleq 10^{-0.6}$. Rearranging \eqref{eq:0113_5} gives 
$
	\frac{\Delta f_t^2}{a_t^2} + \frac{\Delta f_r^2}{a_r^2} = 1$,
which is an ellipse. Hence, the mainlobe width, defined as the length of the minor axis of this ellipse, is given by
\begin{equation}\label{eq:0409_1}
	\mathrm{W}_{\mathrm{main}}
	=
	\min\left\{2\sqrt{({1-\rho_6})/{\sigma_t^2}},\,2\sqrt{({1-\rho_6})/{\sigma_r^2}}\right\}.
\end{equation}
This completes the proof.
\vspace{-2mm}
\bibliographystyle{IEEEtran} 
\bibliography{ref2.bib} 
\end{document}